\documentclass[conference]{IEEEtran}
\IEEEoverridecommandlockouts
\usepackage{times} 

\usepackage{dsfont}
\usepackage{subfig}
\usepackage[inline]{enumitem}
\usepackage{amsmath}
\usepackage{amsfonts}
\usepackage{amssymb}
\usepackage{graphicx}
\usepackage{color}

\UseRawInputEncoding

\usepackage{url}
\usepackage[boxed,vlined]{algorithm2e}
\usepackage{paralist} 
\usepackage{listings} 
\usepackage{bbding} 
\usepackage[draft]{changes} 
\definechangesauthor[name=Christos, color=orange]{CF}

\usepackage{tikz}

\usepackage{diagbox}
\usepackage{xspace}

\newtheorem{definition}{Definition}
\newtheorem{lemma}{Lemma}

\newcommand{\hide}[1]{}

\newcommand{\bit}{\begin{compactitem}}
\newcommand{\eit}{\end{compactitem}}
\newcommand{\ben}{\begin{compactenum}}
\newcommand{\een}{\end{compactenum}}

\newcommand{\method}{\textsc{AutoAudit}\xspace}
\newcommand{\methodsec}{\textsc{AA-Smurf}\xspace}
\newcommand{\methodthi}{\textsc{AA-Discover}\xspace}
\newcommand{\methodfor}{\textsc{AA-Focus}\xspace}

\newenvironment{proof}{\emph{Proof.}}{\hfill$\blacksquare$}

\begin{document}

\title{\method: Mining Accounting and Time-Evolving Graphs}


\makeatletter
\newcommand{\newlineauthors}{%
  \end{@IEEEauthorhalign}\hfill\mbox{}\par
  \mbox{}\hfill\begin{@IEEEauthorhalign}
}
\makeatother

\author{
\IEEEauthorblockN{Meng-Chieh Lee}
\IEEEauthorblockA{National Chiao Tung University \\
jeremy08300830.cs06g@nctu.edu.tw}
\and
\IEEEauthorblockN{Yue Zhao}
\IEEEauthorblockA{Carnegie Mellon University \\
zhaoy@cmu.edu}
\and
\IEEEauthorblockN{Aluna Wang}
\IEEEauthorblockA{Carnegie Mellon University \\
aluna@cmu.edu}
\and
\IEEEauthorblockN{Pierre Jinghong Liang}
\IEEEauthorblockA{Carnegie Mellon University \\
liangj@andrew.cmu.edu}
\newlineauthors
\IEEEauthorblockN{Leman Akoglu}
\IEEEauthorblockA{Carnegie Mellon University \\
lakoglu@andrew.cmu.edu}
\and
\IEEEauthorblockN{Vincent S. Tseng}
\IEEEauthorblockA{National Chiao Tung University \\
vtseng@cs.nctu.edu.tw}
\and
\IEEEauthorblockN{Christos Faloutsos}
\IEEEauthorblockA{Carnegie Mellon University \\
christos@cs.cmu.edu}
}





\maketitle

\begin{abstract} 


How can we spot money laundering in large-scale graph-like accounting datasets? How to identify the most suspicious period in a time-evolving accounting graph? What kind of accounts and events should practitioners prioritize under time constraints? To tackle these crucial challenges in accounting and auditing tasks, we propose a flexible system called \method, which can be valuable for auditors and risk management professionals. To sum up, there are four major advantages of the proposed system:
(a) {\em ``Smurfing'' Detection}, spots nearly 100\% of injected money laundering transactions automatically in real-world datasets.
(b) {\em Attention Routing}, attends to the most suspicious part of time-evolving graphs and provides an intuitive interpretation.
(c) {\em Insight Discovery}, identifies similar month-pair patterns proved by ``success stories'' and patterns following Power Laws in log-logistic scales.
(d) {\em Scalability and Generality}, ensures \method scales linearly and can be easily extended to other real-world graph datasets.
Experiments on various real-world datasets illustrate the effectiveness of our method. To facilitate reproducibility and accessibility, we make the code, figure, and results public at \url{https://github.com/mengchillee/AutoAudit}.

\end{abstract}

\section{Introduction}
\label{sec:intro}


Given a complicated accounting dataset, how can we spot the most suspicious structures and provide practical advice to accountants? Given the graphs rendered from accounting datasets are always directed, weighted, and time-evolving, how can we identify the time frames that exhibit potential correlations and provide the cause? How can we filter out useless information and focus on important signals? We propose \method (AA) to address three major challenges on mining such time-evolving accounting graphs.

The first challenge relates to facilitating the investigation of money laundering crimes (MLCs)--the importance of leveraging machine learning techniques has been recognized in numerous literature \cite{chadha2018handling, paula2016deep, savage2016detection, liflowscope}.
Among all money laundering crimes, ``Smurfing'' \cite{reuter2003money} is one of the most frequent cases \cite{madinger2011money, corselli2020italy, schneider2004money}. As shown in Figure~\ref{fig:smurfing}, it involves the use of multiple intermediaries for making small cash deposits, buying monetary instruments, or bank drafts in amounts under the reporting threshold. Establishing reasonable suspicion usually involves reading and analyzing a massive volume of transactions and thousands of documents, which is exceptionally challenging and time-consuming. Consequently, there is necessity to automate the process of reviewing financial records so that ``Smurfing'' type frauds can be spotted. We propose \methodsec, a parameter-free and unsupervised method, to facilitate the detection of ``Smurfing'' in the accounting data by minimizing the description length of adjacency matrix.

\begin{figure}[!t]
\centering
\includegraphics[width=0.32\textwidth]{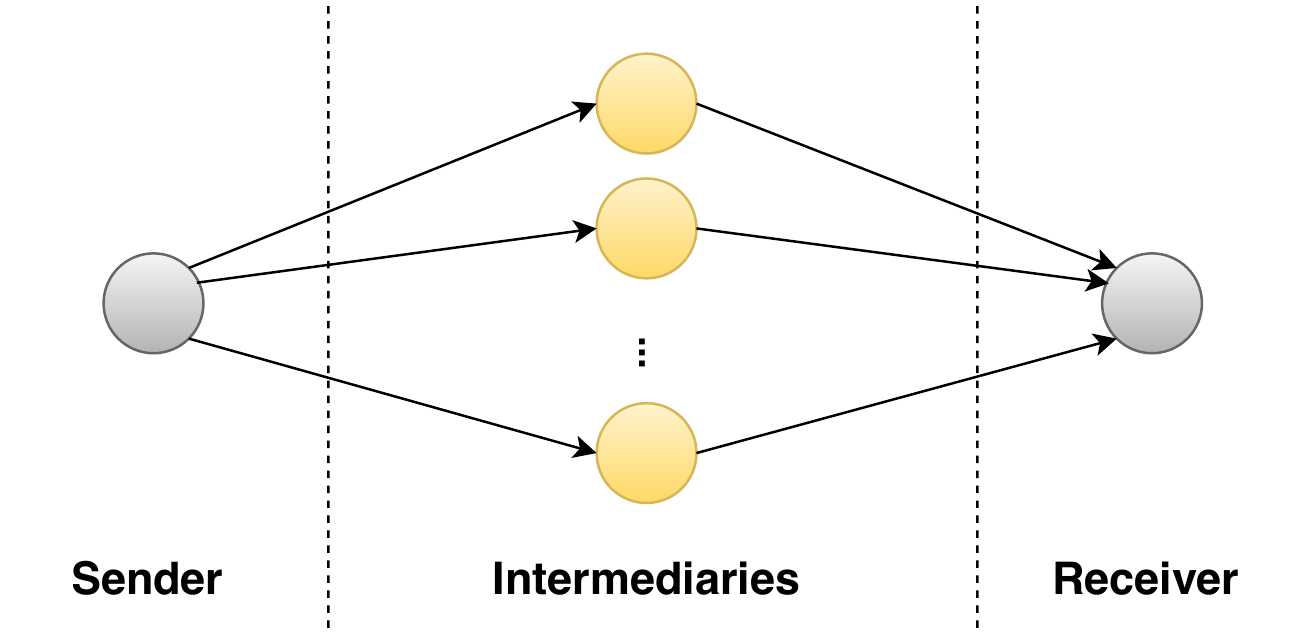}
\vspace{-0.1in}
\caption{\footnotesize {\bf ``Smurfing''}: One of the most frequent types of money laundering}
\label{fig:smurfing}
\end{figure}

The second challenge concerns finding the starting point for the ``test of transactions.'' Due to the sheer volume of financial transactions, limited resources for audits, and ever-changing business environment and policies, auditors often find it challenging to determine a starting point for review when examining transactions and documentation. Although there have been studies in detecting anomalies on time-evolving graphs \cite{ranshous2015anomaly, eswaran2018spotlight, DBLP:conf/kdd/ManzoorMA16, akoglu2010event, eswaran2018sedanspot}, most of them do not offer this information with intuitive explanations. As a result, effective tools that can process account behavior across time to flag critical time-periods for further investigation are needed. We propose \methodfor, by encoding the anomaly score of accounts in each focus-plot over time into sketches, to direct users to the most critical changing points in a long period and explain why they are suspicious. 

The third challenge is discovering the potential properties of time-evolving graphs.
Obtaining and understanding these temporal-correlation patterns can be extremely useful since periodic financial bookkeeping data, visualized in graphs, are often used for internal control purposes. 
We propose \methodthi, to not only bring these temporal patterns to light, but also to ferret out the most effective contributing factors.
Moreover, by using \methodthi, we discover a whole new phenomenon that the log-logistic can fit well for many accounting datasets.
For space limitation, \methodthi along with experimental results are presented in Appendix.



To sum up, we propose \method in this paper---a systematic method for detecting anomalies not only in accounting datasets, but also in other real-world datasets that shares similar characteristics. Figure ~\ref{fig:f1} shows \methodsec automatically reorder the adjacency matrix and put the most likely ``Smurfing'' patterns in the front (highlighted in red).
Moreover, as shown in Figure ~\ref{fig:f3}, \methodfor not only detects the exact period when a significant change occurs, it also finds out why Enron's CEO has been busy solving emergent events.


The advantages of our method are:
\bit
\item {\bf ``Smurfing'' Detection}: \methodsec, an unsupervised and parameter-free algorithm, can detect injected ``Smurfing'' pattern in real-world datasets.
\item {\bf Attention Routing}: \methodfor identifies most suspicious periods in time-evolving graphs with explanations.
\item {\bf Insight Discovery}: \methodthi unearths three month-pairs with high correlation, proved by ``success stories'', and additional patterns of accounting datasets. It is put into Appendix due to space limitation.
\item {\bf Scalability and Generality}: \method scales linearly and can be generalized on other real-world graph datasets, such as Enron Email and Czech Financial datasets.
\eit
{\bf Reproducibility}: Our implemented source code and ``Smurfing'' generator are made public\footnote{\url{https://github.com/mengchillee/AutoAudit}}.


The rest of this paper is organized as follows. First, We introduce the problem statement and detail our proposed method in Section~\ref{sec:meth}. Then, experimental results are then presented in Section~\ref{sec:exp}. Section~\ref{sec:concl} concludes. In Appendix, we briefly review the background and related work in Section~\ref{sec:background}, and provide additional experiment results.

\begin{figure}[!htb]
\centering
\includegraphics[width=0.45\textwidth]{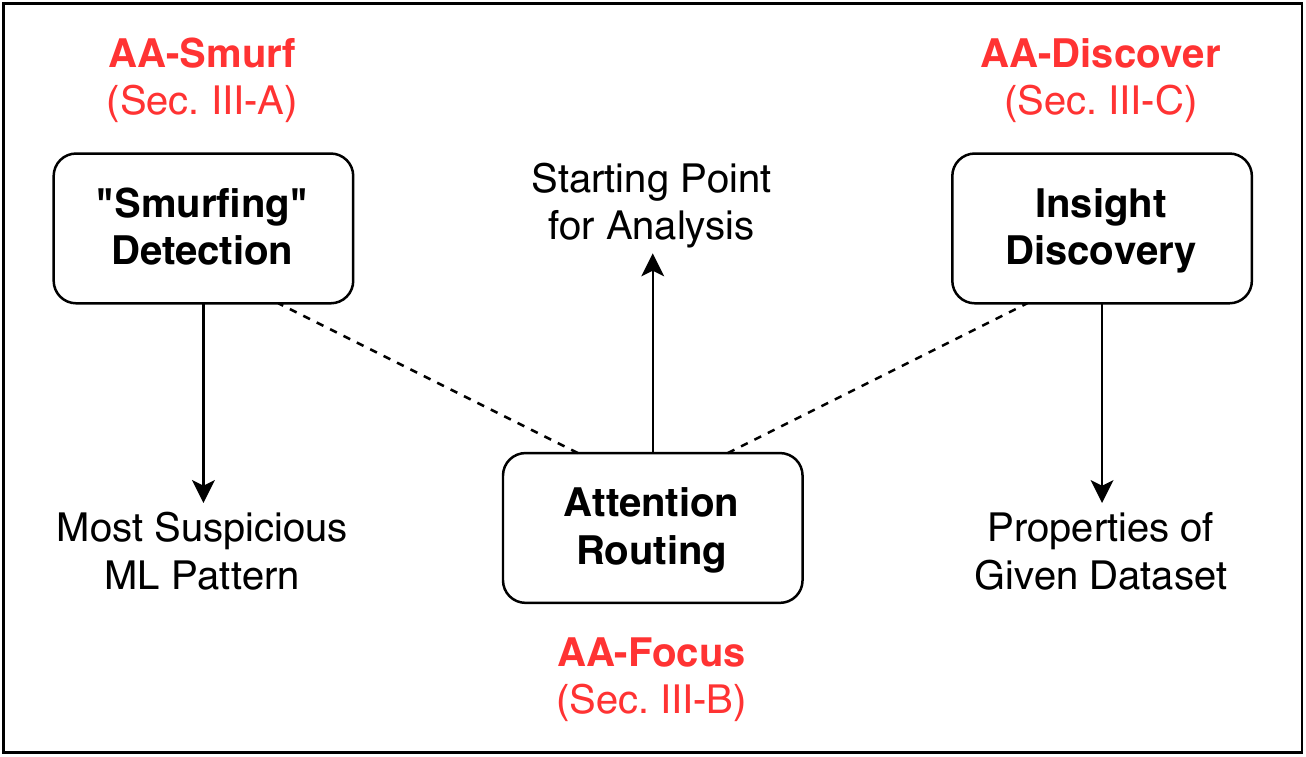}
\vspace{-0.1in}
\caption{{\bf \method}: A systematic overview of \method}
\label{fig:framework}
\end{figure}

\begin{figure*}[!t]
\centering
\subfloat[\methodsec detects and spots ``Smurfing'' successfully]{\label{fig:f1}\includegraphics[scale=0.95]{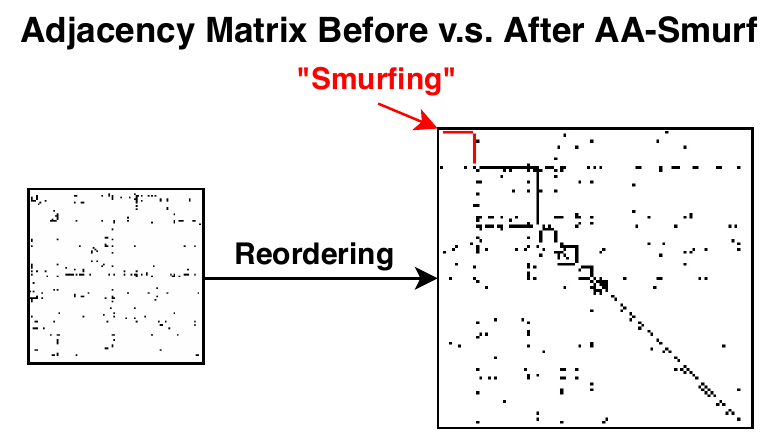}} \qquad
\subfloat[\methodfor detects the time of severe change and give the explanation ]{\label{fig:f3}\includegraphics[scale=0.85]{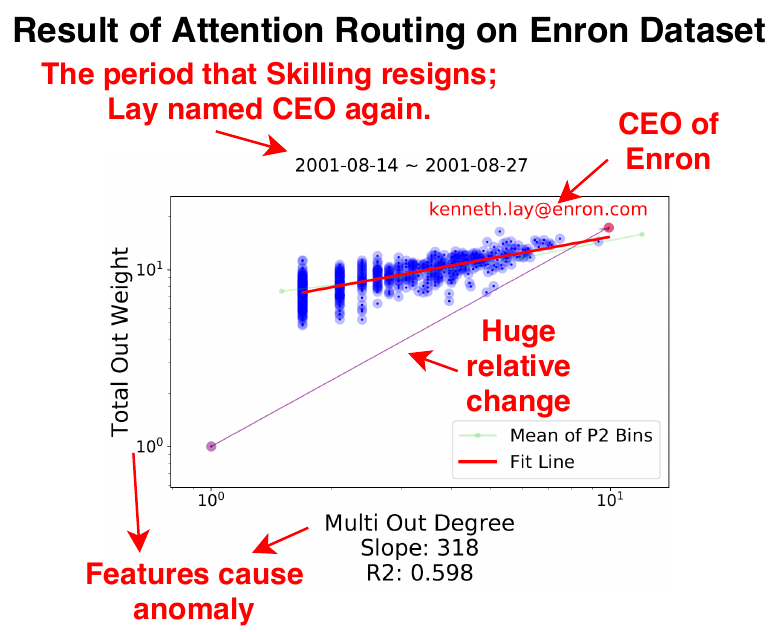}}
\caption{\method contains \methodsec (left) to detect ``Smurfing'' and spots by matrix reordering, and \methodfor (right) to detect the severe event in Enron dataset and explains by the 2-D plot with huge relative change.}
\label{fig:summary}
\end{figure*}

\section{Proposed Method}
\label{sec:meth}

As shown in Figure~\ref{fig:framework}, \method solves these three major challenges.
We present the details of \methodsec on detecting ``Smurfing'' in Section~\ref{ssec:num1}, and then the design of attention routing in Section~\ref{ssec:num3}. Finally, discovering the properties of given dataset is discussed in Appendix~\ref{app:num2}. 

\subsection{\methodsec for ``Smurfing'' Detection} \label{ssec:num1}
\subsubsection{Problem Definition}
We focus on the first placement stage of ``Smurfing'', and define the problem as follows:
\begin{definition}[``Smurfing'']
Given a time-evolving accounting dataset without ground truth, the goal is to detect ``Smurfing'', where there are one sender, one receiver, and $k$ intermediaries as shown in Figure~\ref{fig:smurfing}. The task should render results understandable to auditors and risk management professionals.
\end{definition}
There are three basic properties of ``Smurfing'' derived from the information in \cite{levi2006money}, which can be used for detection:
\bit
\item {\bf Intermediary Size:} The more the intermediaries, the more suspicious.
\item {\bf Inside Purity:} The less the interactions between intermediaries, the more suspicious.
\item {\bf Outside Purity:} The less the intermediaries connected or connecting to normal accounts, the more suspicious.
\eit
Although intermediaries usually camouflage themselves by setting up connections with normal accounts, they try their best to avoid interacting with other intermediaries. If there exists interactions between intermediaries, it is hard to conceal information and transactions. Moreover, once one of the intermediaries is exposed, it is much easier to track others.


\subsubsection{Overall Algorithm}
Algorithm~\ref{algo:robotacct} gives the steps of \methodsec. In line 1, we first collect all possible index sets of sender, intermediaries and receiver $\mathcal{P}$, if the number of intermediaries is higher than 2. In line 3-17, we iteratively identify the best ``Smurfing'' patterns. The score function is given in line 9, consisting of purity and encoding cost. In each iteration, all the sets $I^{'} \in \mathcal{P}$ are examined to find the highest score $s$ in line 4-10. If no pattern generating lower cost is to be found, in line 18, we identify the final result by picking the subset $\mathcal{I}^{*}$ of $\mathcal{I}$, Index sets containing $\{I_1, I_2, ...\}$, with encoding cost 10\% greater than the minimum (early stopping) as described in \cite{DBLP:conf/sigmod/MatsubaraSF14}. We then detail each of the subroutines in the following paragraphs.

\begin{algorithm}[htbp]
\SetAlgoVlined
\SetAlgoLined
\LinesNumbered
\KwData{Adjacency matrix $V$ from graphs $G$}
\KwResult{Reordered matrix with highest score}
$\mathcal{P} = GetSets(V), \mathcal{I} = []$\;
$c = Cost_{T}(V, [])$\;
\While{True}{
  \For{$I^{'} = (i^{s}, i^{m}, i^{r}) \in \mathcal{P} \backslash \mathcal{I}$}{
        $\mathcal{I}^{'} = \mathcal{I} \cup {I^{'}}$\;
        /* Matrix Reordering \S~\ref{par:matre} */ \\
        $V^{'} = Reorder(V, \mathcal{I}^{'})$ \;
        /* Encoding Cost \S~\ref{par:encoding} */ \\
        $c^{'} = Cost_{T}({V}^{'}, \mathcal{I}^{'})$ \;
        /* Purity Calculation \S~\ref{par:purity} */ \\
        $p^{'} = Purity({V}^{'}, \mathcal{I}^{'})$ \;
        $s = p^{'} * (c - c^{'}) / c$\;
  }
  \If{No set with cost lower than $c$}{
      Break\;
  }
  Pick the best $V^{*}$ and $I^{*}$ from $\mathcal{P}$ with the largest $s$\;
  Add $I^*$ into $\mathcal{I}$\;
  $c = Cost_{T}({V}^{*}, \mathcal{I})$\;
}
Pick $\mathcal{I}^{*} \in \mathcal{I}$ with cost 10\% greater than $c$\;
Return $Reorder(V, \mathcal{I}^{*})$\;
\caption{\methodsec: Reorder the adjacency matrix and display the most suspicious patterns \label{algo:robotacct}}
\end{algorithm}

\subsubsection{Matrix Reordering} \label{par:matre}
Our method is based on the intuition that if an adjacency matrix of graph contains ``Smurfing'' patterns, after reordering, it is much easier to be compressed. Based on the three aforementioned properties, the optimal matrix containing one ``Smurfing'' pattern should be as follows:
\begin{equation}
    V_{i, j} =
    \begin{cases}
        1 & \text{if } i = 0 \text{ and } j \in [1, k - 1] \\
        1 & \text{if } j = k \text{ and } i \in [1, k - 1] \\
        0 & \text{otherwise}
    \end{cases}
    i, j \in [0, n]
\end{equation}
where $n$ is the number of accounts and $k$ is the number of intermediaries.
An illustration is shown by the sub-matrix $A$ in Figure~\ref{fig:mtr}, where it only contains totally $2k$ elements. The matrix orders of row and column are the same in our method. As shown in Algorithm~\ref{algo:reorder}, given a matrix $V$ at iteration $J$, where $j = 1, ..., J$, we add the indexes of sender, intermediaries and receiver into order $\mathcal{O}$ respectively in line 2-6 iteratively. We then reorder the row and column in $V$ by $\mathcal{O}$ in line 8. 

\begin{algorithm}[htbp]
\SetAlgoVlined
\SetAlgoLined
\LinesNumbered
\KwData{Matrix $V$, and index set of sender, intermediaries and receiver $\mathcal{I}$}
\KwResult{Reordered matrix $V^{'}$}
Initialize order $\mathcal{O} = []$ \;
\For{$I_j = (i^s_j, i^m_j, i^r_j) \in \mathcal{I}$}{
    Add sender index $i^s_j$ to $\mathcal{O}$\;
    Add intermediaries indexes $i^m_j$ to $\mathcal{O}$\;
    Add receiver index $i^r_j$ to $\mathcal{O}$\;
}
Add the rest indexes into $\mathcal{O}$\;
$V^{'} = $ reorder the rows and columns of $V$ by $\mathcal{O}$\;
Return $V^{'}$\;
\caption{Reorder: Reorder the matrix at iteration $J$ by the indexes of sender, intermediaries and receiver \label{algo:reorder}}
\end{algorithm}

\subsubsection{Encoding Cost of Matrix Reordering} \label{par:encoding}
To encode the reordered matrix, we use Minimum Description Length (MDL) principle \cite{rissanen1978modeling}. It finds the best way to compress matrices based on the aforementioned properties of ``Smurfing'' patterns. Based on this basic idea, we design a method to compress the binary adjacency matrix. As shown in Figure~\ref{fig:mtr}, given a reordered adjacency matrix $V_J$ and index information $\mathcal{I}$ at iteration $J$, we denote intermediaries number as $k_j = |i^{m}_j|$ and start indexes as $r_j = \sum_{j^{'}=1}^{j-1} |I_{j^{'}}|$ of ``Smurfing'' patterns, and separate the reordered adjacency matrix into sub-matrices as follows:
\begin{equation}
\begin{array}{l}
    A_j(V_J, I_j) = V_J[r_j:r_j+k_j][r_j:r_j+k_j] \\
    B_j(V_J, I_j) = V_J[r_j:r_j+k_j][r_j+k_j+1:n] \\
    C_j(V_J, I_j) = V_J[r_j+k_j+1:n][r_j:r_j+k_j]
\end{array}
\end{equation}
where $j = 1, ..., J$ and $I_j \in \mathcal{I}$. Sub-matrix $A_j$ denotes the potential ``Smurfing'' pattern separated out by reordering. Sub-matrix $B_j$ denotes intermediaries sending money, and sub-matrix $C_j$ denotes intermediaries receiving money. Sub-matrix $D$ is used to denote errors which cannot be compressed by found patterns, and is defined as follows:
\begin{equation}
    D(V_J, I_J) = V_J[r_J+k_J+1:n][r_J+k_J+1:n]
\end{equation}
Based on the inside purity and outside purity properties, in order to have fewer interactions, the encoding cost of sub-matrices $A$, $B$ and $C$ can be computed as follows:
\begin{equation}
\begin{array}{l}
    Cost_{A}(V_J, \mathcal{I}_J) = (\sum_{j=1}^{J}{|A_j(V_J, I_j)| - 2k_j) * (2\lg{(k_j - 1)})} \\
    Cost_{B}(V_J, \mathcal{I}_J) = (\sum_{j=1}^{J}{|B_j(V_J, I_j)|) * (\lg{n} + \lg{k_j})} \\
    Cost_{C}(V_J, \mathcal{I}_J) = (\sum_{j=1}^{J}{|C_j(V_J, I_j)|) * (\lg{n} + \lg{k_j})}
\end{array}
\end{equation}
This function encodes the position of non-zero elements in the given matrix. We subtract $2k_j$ for each $A_j$ because there is no need to record the occurrence of intermediaries. To ensure that most zeros can be compressed in $A_j$, $B_j$, and $C_j$, sub-matrix $D$ should contain as few zeros as possible, which can be calculated as follows:
\begin{equation}
    Cost_{D}(V_J, \mathcal{I}_J) = (|\mathds{J} - D(V_J, I_J)|) * (2\lg{n})
\end{equation}
where $\mathds{J}$ is a matrix of ones. This means that instead of encoding the positions of non-zeros, we encode the positions of zeros in the sub-matrix $D_j$. To encode real numbers, we use universal code length $lg^{*}(x)\approx2lg(x)+1$. The number of found patterns is encoded by $lg^{*}(J)$. The indexes for each sender, receiver and intermediaries need $lg(n)$ for each to be encoded. The total cost can then be formed as follows:
\begin{equation}
\begin{array}{l}
    Cost_{T}(V_J, \mathcal{I}_J) = Cost_{A}(V_J, \mathcal{I}_J) + Cost_{B}(V_J, \mathcal{I}_J) \\
    \quad  + Cost_{C}(V_J, \mathcal{I}_J) + Cost_{D}(V_J, \mathcal{I}_J) \\ 
    \quad  + lg^{*}(J) + |\mathcal{I}_J| * \lg{n}
\end{array}
\end{equation}
To compute the final score, we use compression rate instead of the actual cost. The compression rate is calculated by normalizing the difference of original and compressed MDL by the original MDL, i.e, $(c-c^{'})/c$, where $c$ denotes the encoding cost of the last iteration and $c^{'}$ denotes the current encoding cost.

\begin{figure}[!t]
\centering
\includegraphics[width=0.22\textwidth]{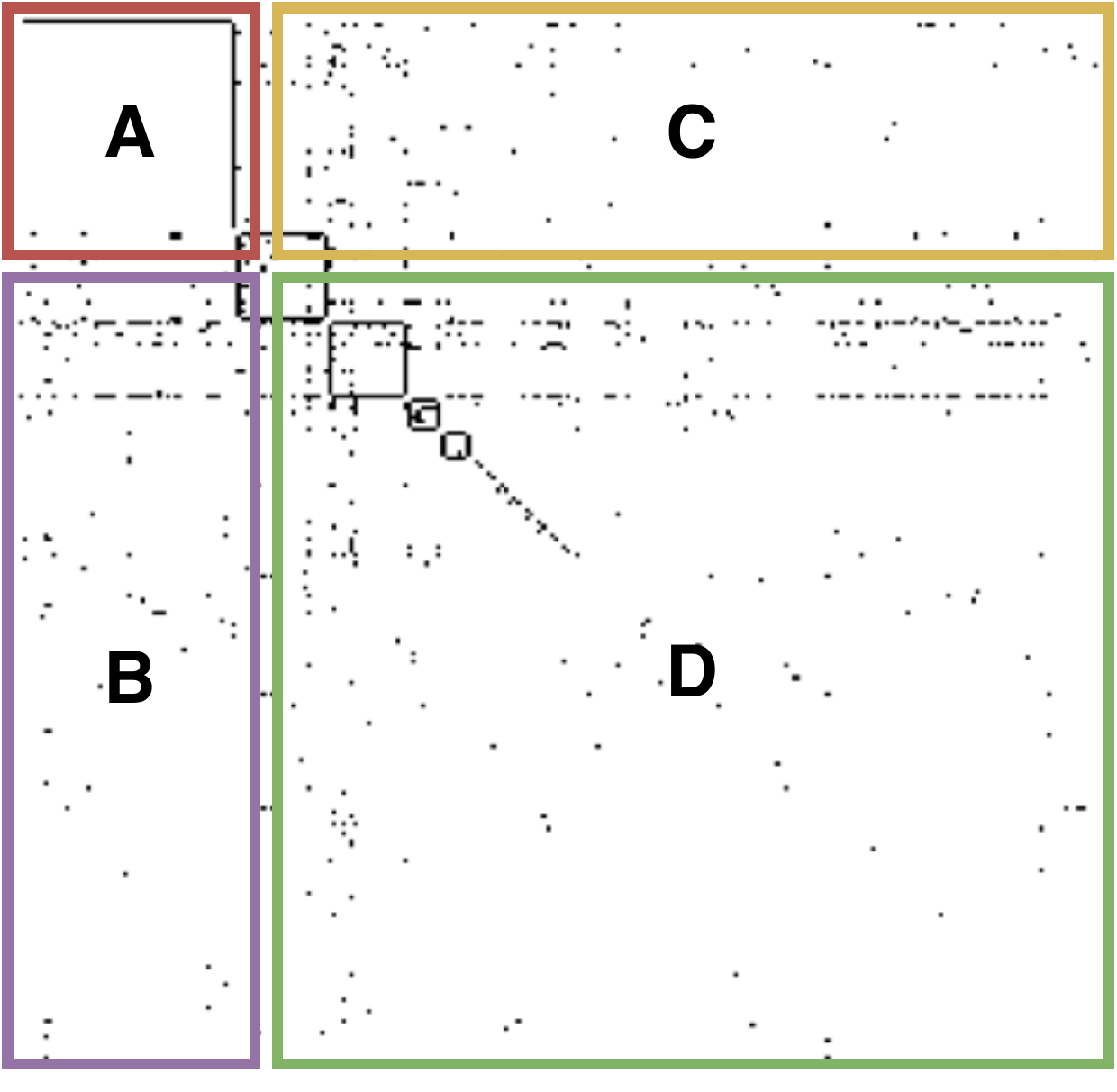}
\vspace{-0.1in}
\caption{Reordered binary adjacency matrix with ``Smurfing'', where sub-matrix A denotes ``Smurfing'' and the inside purity, sub-matrix B and C denote the outside purity, and sub-matrix D denotes the errors}
\label{fig:mtr}
\end{figure}

\subsubsection{Purity Calculation} \label{par:purity}
It is obvious that if the numbers of intermediaries increase, the matrix tends to have a higher compression rate. When only using compression rate to select an anomaly, the one with more intermediaries will have a higher chance to be selected, even when its purity is low. To solve this issue, we develop a new score called purity. At iteration $J$, the purity score $p$ can be computed as follows:
\begin{equation}
    {\textstyle Purity(V_J, \mathcal{I}_J) = \sum_{j=1}^{J}{\frac{2k_j}{|A_j(V_J, I_j)| + |B_j(V_J, I_j)| + |C_j(V_J, I_j)|}}}
\end{equation}
where $0 < p \leq 1$. In the best case, where $p=1$, the edges will only connect from the sender to intermediaries, and from intermediaries to the receiver. We then multiply the compression rate by purity to get the final score $s$ of a given index set.


\subsubsection{Complexity Analysis}
\begin{lemma}
\methodsec takes time
\begin{equation}
    O(nz+|\mathcal{P}|^{2})
\end{equation}
where $z$ denotes the number of non-zero elements and $|\mathcal{P}|$ denotes the number of possible index sets.
\end{lemma}

\begin{proof}
By using the sparse matrix computation, the time complexity of obtaining possible index sets can be reduced to $O(nz)$, where $z$ denotes the number of non-zero elements. The worst case of the while loop in the algorithm is that every iteration we only remove one set from $\mathcal{P}$, taking $O(\frac{(|\mathcal{P}| + 1)\mathcal{P}}{2})$ to run (sum of an arithmetic progression). In practice, the number of iterations is linear to the input size, since each account cannot be reordered twice. If an account is identified as an intermediary, the possible sets involving it are removed. The complexity of reordering is $O(|\mathcal{I}|)$, where $|\mathcal{I}|$ is usually a small constant that is negligible. Thus, the complexity is $O(nz+|\mathcal{P}|^{2})$.
\end{proof}

\subsection{\methodfor for Attention Routing} \label{ssec:num3}

\subsubsection{Problem Definition}
In a huge graph dataset, finding the starting point for analysis can be extremely difficult. It becomes even more challenging when the dataset is dynamic, since neither analyzing journal entries entered at a specific time point nor the whole time span would be effective and efficient for detecting changes occurring only in a short period. Furthermore, it is equally important to offer viable explanations for these changes after detecting them. One thing worth noting is that these changes may not always be anomalies, they may also be related to other types of important events. Here we define the Attention Routing problem as follows:

\begin{definition}[Attention Routing]
Given time-evolving graph datasets, the goal is to detect sudden changes over time, and use the least resources to explain the most patterns.
\end{definition}

\subsubsection{Graph Construction} \label{par:graphcon}
We first split the time-evolving graph $G$ into $I$ graphs $G_i$ over time by sliding windows with overlap. In Algorithm~\ref{algo:aaar} line 1-4, we construct bipartite graphs $W_i$ for each graph $G_i$, where the source nodes on the left side denote each node in $G_i$ and the destination nodes on the right side denote explainable focus-plots. We focus on drawing visualizable and understandable 2-d focus-plots, which scatters each data point by only two features. There are total $\left(\begin{array}{c}q\\ 2\end{array}\right)$ focus-plots generated from the features of nodes, where $q$ denotes the number of extracted features from nodes. In each focus-plot, the anomaly score of each node is calculated by an off-the-shelf unsupervised anomaly detection model. We use Isolation Forest \cite{liu2008isolation} because of its efficiency and robustness. The edge weight $w^i$ of the bipartite graph $W_i$ is then set as the anomaly score of nodes for each focus-plot.

\begin{algorithm}[htbp]
\SetAlgoVlined
\SetAlgoLined
\LinesNumbered
\KwData{Graphs $G$ split from sliding window}
\KwResult{Time $t_a$ with the highest change score and $v$ 2-d focus-plots}
/* Graph Construction \S~\ref{par:graphcon} */ \\
\For{$G_i \in G$}{
    Extract node feature pairs $f_i$ from $G_i$ \;
    Construct bipartite graph $W_i$\;
}
/* Sketch Generation \S~\ref{par:sketchgen} */ \\
\For{$j = 1, ..., l$}{
    Randomly select $\mathcal{S}_j$ and $\mathcal{D}_j$\;
    Select $v$ focus-plots $\mathcal{V}_j$ from $\mathcal{D}_j$ by Eq.~\ref{eq:max}\;
}
Generate $l$-d sketch $K$ for all $\mathcal{S}$ and $\mathcal{V}$\;
/* Change Point Detection and Attention \S~\ref{par:changepoint} */ \\
\For{$i = t, ..., len(G)$}{
    Compute principal eigenvector $e_i$ by SVD from $K_{i-t:i-1}$\;
    Compute cosine distance between $e_i$ and $K_i$ as change score over time\;
}
Return $t_a$ with maximum change score and $v$ focus-plots\;
\caption{\methodfor: Attend to the most suspicious time period and explain by the focus-plots \label{algo:aaar}}
\end{algorithm}

\subsubsection{Sketch Generation} \label{par:sketchgen}
Once we have bipartite graphs for each timestamp, to detect change in the anomaly scores of each node, each bipartite graph is encoded into sketches to apply further analysis. We borrow the idea from SpotLight \cite{eswaran2018spotlight}. $l$ subsets of source $\mathcal{S}$ and $l$ subsets of destination $\mathcal{D}$ of bipartite graphs have been randomly selected with a given probability, where $l$ denotes the dimensions of sketches. The sketch value in dimension $j$ can then be calculated as follows:
\begin{equation}
    f(\mathcal{S}_j, \mathcal{D}_j) = \sum_{s_x \in \mathcal{S}}{\max_{d_y \in \mathcal{D}}{\sum_{i=1}^{I}{w^i_{s_x, d_y}}}}
\end{equation}
where $f$ denotes the sum of the edge with maximum weight linking from source subset to corresponding destination subset in all $W$. The total weight of edge $(s_x, d_y)$ in in all $W$ gives a global view of how well this focus-plot $d_y$ explaining $s_x$.

The difference between our method and SpotLight is that we improve the selection of destination subsets. After randomly selecting destinations (focus-plots) of sketch $K_j$, we aim to find the subset of $\mathcal{D}_j$, denoted as $\mathcal{V}_j$, which maximally explains the subset of source (nodes) $\mathcal{S}_j$. The gains of adding each focus-plot is then computed by marginal gains \cite{DBLP:journals/corr/abs-1710-05333} as follows:
\begin{equation}
    \Delta_f(d|\mathcal{V}^i_j) = f(\mathcal{V}^{i-1}_j \cup \{d\}) - f(\mathcal{V}^{i-1}_j)
\end{equation}
As shown in line 5-8 of Algorithm~\ref{algo:aaar}, we iteratively add the destinations that maximize the marginal gain for $v$ iterations:
\begin{equation}\label{eq:max}
    \mathcal{V}^i_j = \mathcal{V}^{i-1}_j \cup \{{\max_{d \in \mathcal{D}_j \backslash \mathcal{V}^{i-1}_j}}{\Delta_f(d|\mathcal{V}^i_j)}\}
\end{equation}
where $i = 1, ..., v$ and $j = 1, ..., l$. This ensures that the $v$ selected focus-plots have the best performance on explaining the nodes.

\subsubsection{Change Point Detection and Attention} \label{par:changepoint}
Having the sketch for each timestamp, we use \cite{akoglu2010event} to measure distance between current and past behavior as the changing behavior of graphs. In Algorithm~\ref{algo:aaar} line 10-13, at timetick $i$, the behavior of past $t$ timeticks can be summarized by computing the principal eigenvector $e$ of the matrix $K_{i-t:i-1}$ sketches by Singular Value Decomposition (SVD). We use cosine distance to measure the difference as the change score, and identify the dimension with the largest relative change. After finding that dimension, we have a subset of nodes and $v$ focus-plots, where $v$ is the number of plots to output, which give the most informative explanation to this changed behavior.

\subsubsection{Complexity Analysis}
\begin{lemma}
\methodfor takes time
\begin{equation}
    O(|G|(q^{2}f\varphi\log{\varphi}+tl))
\end{equation}
where $f$ and $\varphi$ denote the number of trees and the max sample size in Isolation Forest, and $|G|$ is the number of graphs over time.
\end{lemma}

\begin{proof}
Isolation Forest takes $O(f\varphi\log{\varphi})$ to run, where $f$ denotes the number of trees and $\varphi$ denotes the max sample size. Since each focus-plot needs to run Isolation Forest once, the complexity to construct a bipartite graph takes $O(q^{2}f\varphi\log{\varphi})$. For total $|G|$ graphs over time, the complexity of graph construction is $O(|G|q^{2}f\varphi\log{\varphi})$. Moreover, $O(l)$ is needed to generate the sketches, and $O(|G|tl)$ is needed to extract the principled eigenvectors by SVD. Therefore, the complexity of \methodfor is $O(|G|(q^{2}f\varphi\log{\varphi}+tl))$.
\end{proof}

\section{Experiments}
\label{sec:exp}




In this section, we present our experimental results to answer the following questions:
\begin{compactenum}[{Q}1.]
\item {\bf Effectiveness}: How well does \method work on real-world datasets?
\item {\bf Scalability}: How does \method's running time grow with input size?
\item {\bf Discovery}: What is the most interesting thing found by \method? (See Appendix~\ref{app:ar} and \ref{app:id})
\end{compactenum}

We implement \method in Python; experiments are run on a 3.2 GHz CPU with 256 GB RAM on a Linux server.

In our experiments, we use three real-world graph datasets. The Accounting dataset from an anonymous institution contains transactions and accounts from one company, which precisely reflects the money flow inside the company. In addition to the Accounting dataset, we include two more real-world time-evolving graph datasets. The Czech Financial dataset \cite{berka1999workshop} contains transactions between customers (inside the bank) and clients (outside the bank), and there is no interaction between customers and clients. The Enron Email dataset \cite{klimt2004enron} collects emails at Enron from about 1998 to 2002, and we used the data from 2001 for concision in our experiment. None of these datasets come with ground truth. The graphs we used in our experiments are described in the Table~\ref{tab:datasets}.

\begin{table}[ht]
\begin{center}
\begin{tabular}{c|c|c|c}
    \multicolumn{4}{l}{\textbf{Real-world Graph Datasets}} \\ \hline \hline 
    Name & Nodes & Edges & Time Span \\ \hline \hline
    Accounting & 254 & 285,298 & 01/01/2016 to 02/06/2017 \\
    Czech Financial & 11,374 & 273,508 & 01/05/1993 to 12/14/1998 \\
    Enron Email & 16,771 & 1,487,863 & 01/01/2001 to 12/31/2001 \\
  \hline \hline
\end{tabular}
\end{center}
\vspace{-0.1in}
\caption{The summary of real-world graph datasets used}
\label{tab:datasets}
\end{table}


\subsection{``Smurfing'' Detection}

\subsubsection{Setup}
To evaluate \methodsec, we inject anomalies in our Accounting and Czech datasets. In the Accounting dataset, we randomly pick one sender and one receiver, and then pick $n$ intermediaries. The interactions between intermediaries have been deleted to mimic ``Smurfing''. In the Czech dataset, sender and receiver are randomly picked from client accounts, and intermediaries are picked from customer accounts. This imitates ``Smurfing'' where money is imported from outside of the bank and transferred out. The interaction between customer accounts are randomly generated with the probability of 0.05\%. We injected the datasets with intermediaries from 10 to 50 respectively. Noise patterns with either more intermediaries or more interaction between intermediaries are randomly added to test robustness. For each number of intermediaries we run experiments for 10 times to ensure low variance. Because of no existing baseline, so we run an ablation study like evaluation instead. Baseline \methodsec w/o MDL and \methodsec w/o Purity only contain one of the two ingredients in score function (purity or encoding cost). Baseline number of intermediaries directly detects ``Smurfing'' by outputting the pattern with the highest number of intermediaries.

\begin{figure}[!htb]
\centering
\subfloat{\label{fig:EFFACC}\includegraphics[scale=0.36]{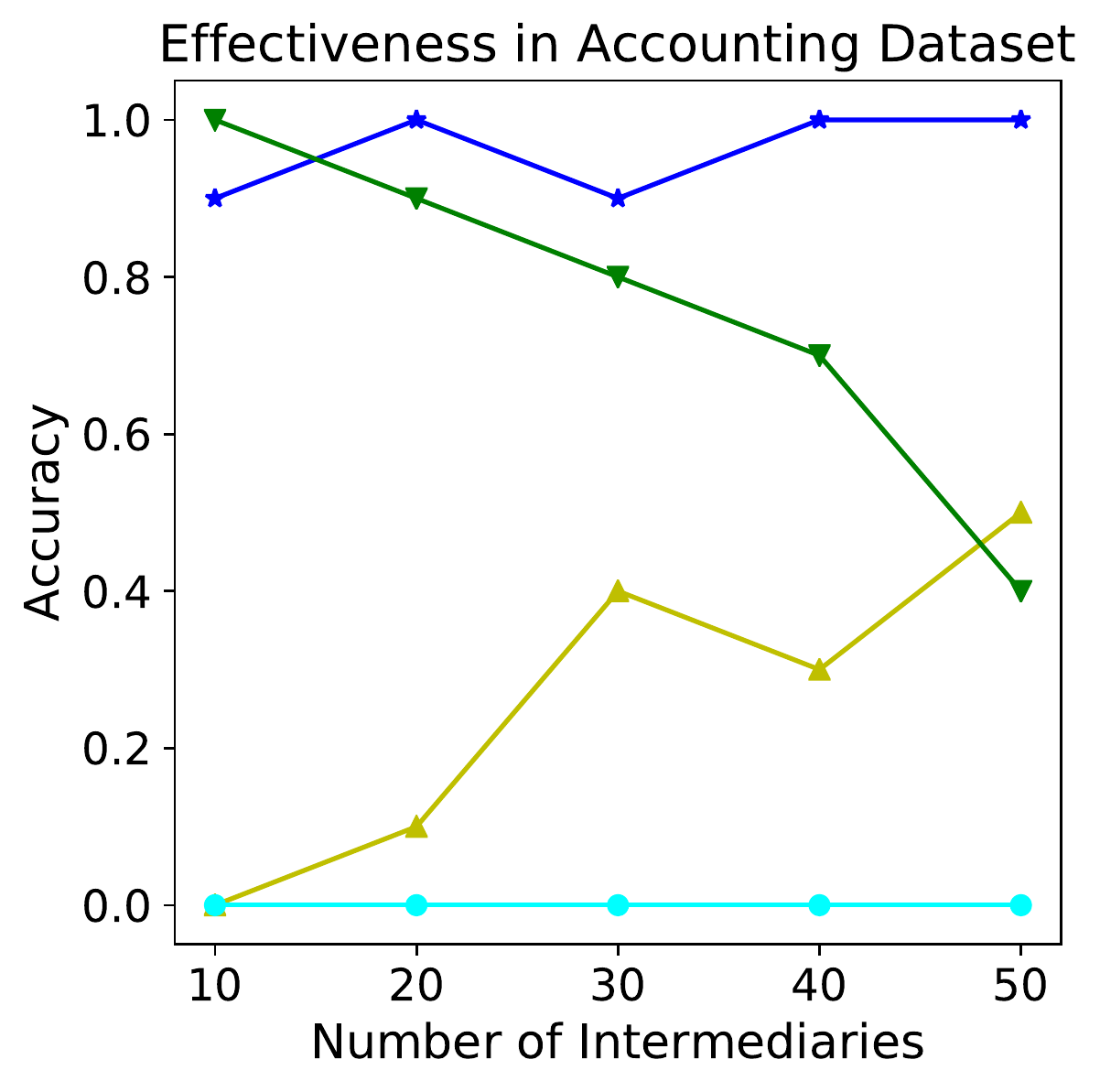}}
\subfloat{\label{fig:EFFCDF}\includegraphics[scale=0.36]{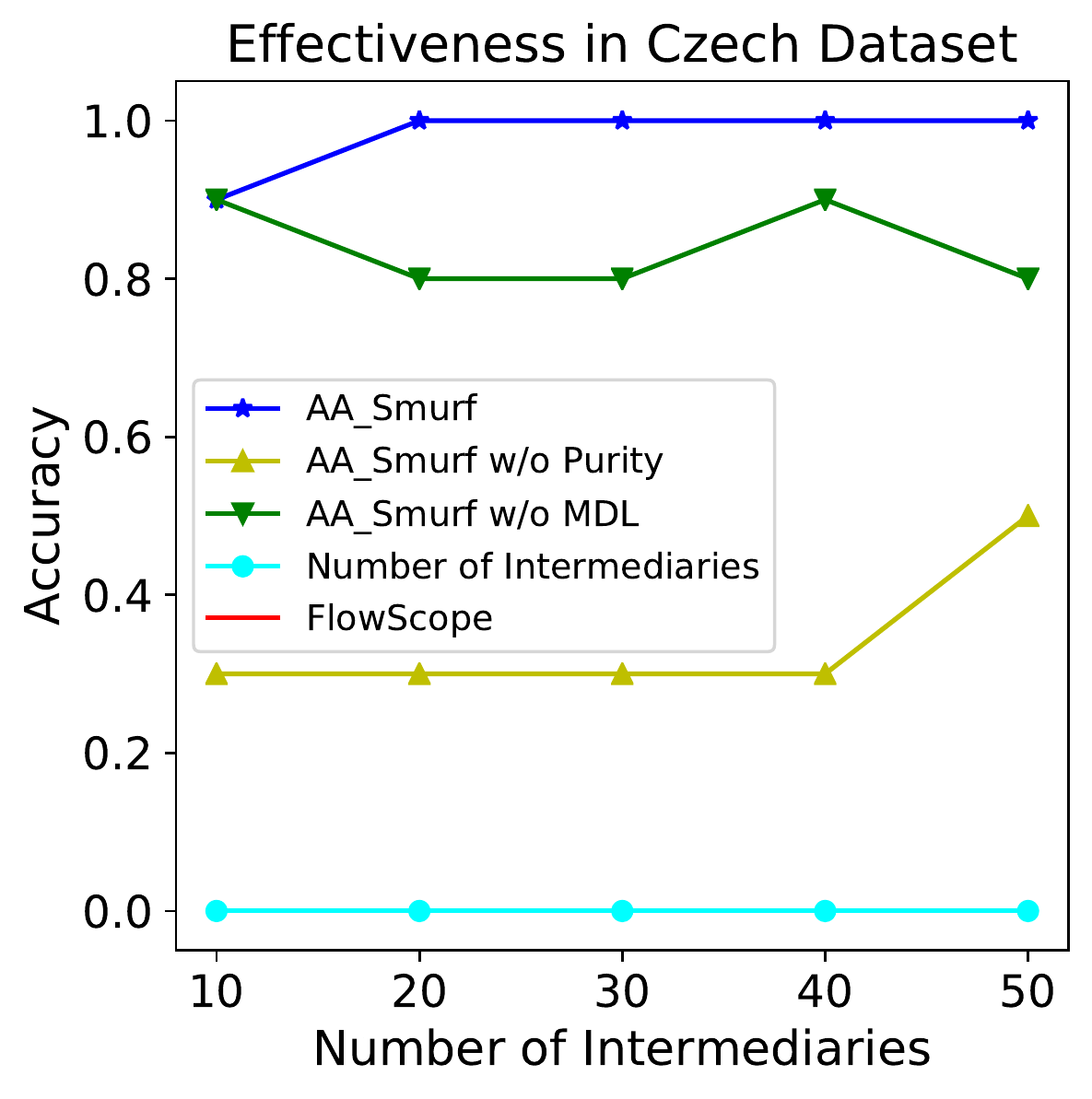}}
\vspace{-0.1in}
\caption{{\bf \methodsec}: \methodsec achieves nearly 100\% accuracy and outperforms other baselines in Accounting Dataset and Czech Dataset}
\end{figure}



\subsubsection{Effectiveness}
As shown in Figure~\ref{fig:EFFACC} and Figure~\ref{fig:EFFCDF}, \methodsec successfully detects ``Smurfing'' patterns in both Accounting and Czech datasets and outperforms other baselines by achieving nearly 100\% accuracy. Baseline number of intermediaries performs the worst, indicating the necessity of designing a ``Smurfing'' detection method. Moverover, FlowScope \cite{liflowscope} is not usable in our problem because they enabled intermediaries can have many interactions. In Figure~\ref{fig:f1}, we also can see \methodsec reorders the matrix by putting the most suspicious patterns in the front, so that they can be more easily noticed by auditors and risk management professionals.




\subsubsection{Scalability}
To evaluate the scalability of \methodsec, we vary the number of possible index sets, which highly correlate with running time. As shown in Figure~\ref{fig:aasmurfsca1} with black triangles, \methodsec scales linearly with the number of possible index sets $\mathcal{P}$ with only 1 machine. We also parallelized it by 4 machines to get a much higher efficiency, shown in Figure~\ref{fig:aasmurfsca1} with red squares.
To demonstrate that \methodsec is highly parallelizable, we vary the number of machines and report $\frac{T_1}{T_M}$, where $T_1$ denotes the running time by 1 machine, and $M$ denotes the number of machines. In Figure~\ref{fig:aasmurfsca2}, we find that it achieves a 4.5 times speedup by only 12 machines.

\begin{figure}[h]
\centering
\subfloat{\label{fig:aasmurfsca1}\includegraphics[scale=0.36]{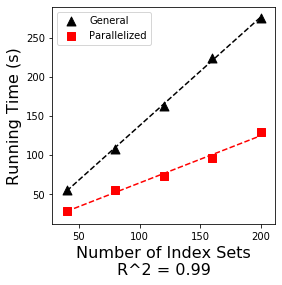}}\qquad
\subfloat{\label{fig:aasmurfsca2}\includegraphics[scale=0.5]{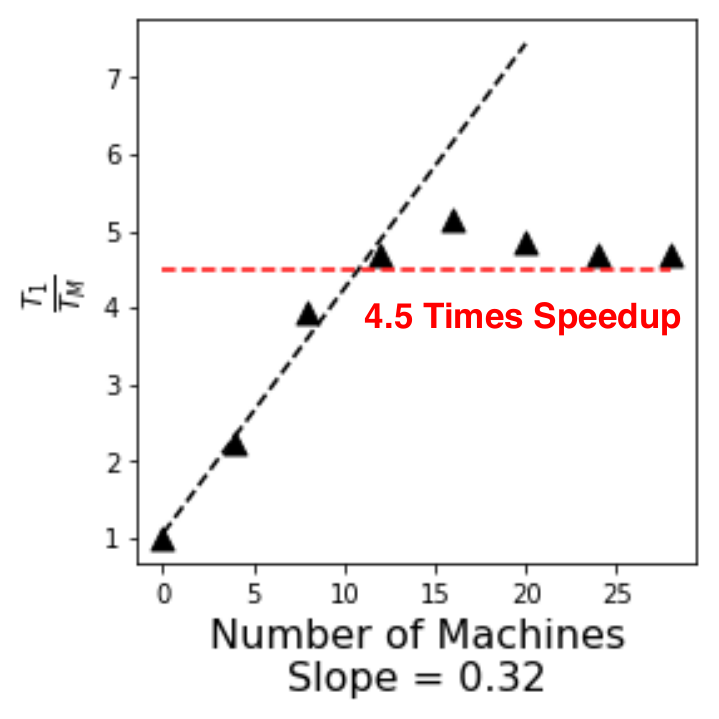}}
\vspace{-0.1in}
\caption{{\bf \methodsec}: It scales linearly with number of index sets by one machine (left), and speeds up if running with more than 1 machine (right).}
\label{fig:aasmurfsca}
\end{figure}


\subsection{Attention Routing}

\subsubsection{Setup}
In this experiment, graphs are separated along time by 14-day time window with 50\% overlap, and then we extract all the following features from graphs: 
\begin{enumerate*}[label=(\roman*)]
  \item Unique In Degree
  \item Multi In Degree
  \item Unique Out Degree
  \item Multi Out Degree
  \item Total In Weight
  \item Mean In Weight
  \item Median In Weight
  \item Variance In Weight
  \item Total Out Weight
  \item Mean Out Weight
  \item Median Out Weight
  \item Variance Out Weight, 
\end{enumerate*}
where multiplicity denotes repetition. These features have been transformed into logarithm scale before being input into Isolation Forest. The dimension of each sketch is set to be 256.

\subsubsection{Effectiveness}
To test our effectiveness, we test \methodfor on the Enron Email dataset, with which we could access accurate events on the news timeline. In Figure~\ref{fig:ar3}, we find that the change score fluctuates over time. The red dot in Figure~\ref{fig:ar3} corresponds to the time when Jeff Skilling resigned after releasing the August 2001 quarterly financial report and Kenneth Lay returned to being the chief executive. In Figure~\ref{fig:f3}, the unique out degree and median out weight of Kenneth Lay increased dramatically in a short time, reflecting this to be an important event. Figure~\ref{fig:ar3} (red box) shows the period that Enron declared bankruptcy. As shown in Figure~\ref{fig:ar4}, Karen Denne, the spokeswoman of Enron, have been busy handling media interviews, explaining the remarkable increase of the total out weight and multi out degree.

\begin{figure}
\centering
\subfloat[The change score over time, where the significant changes are shown in red point and red box]{\label{fig:ar3}
\vspace{-0.1in}
\includegraphics[scale=0.23]{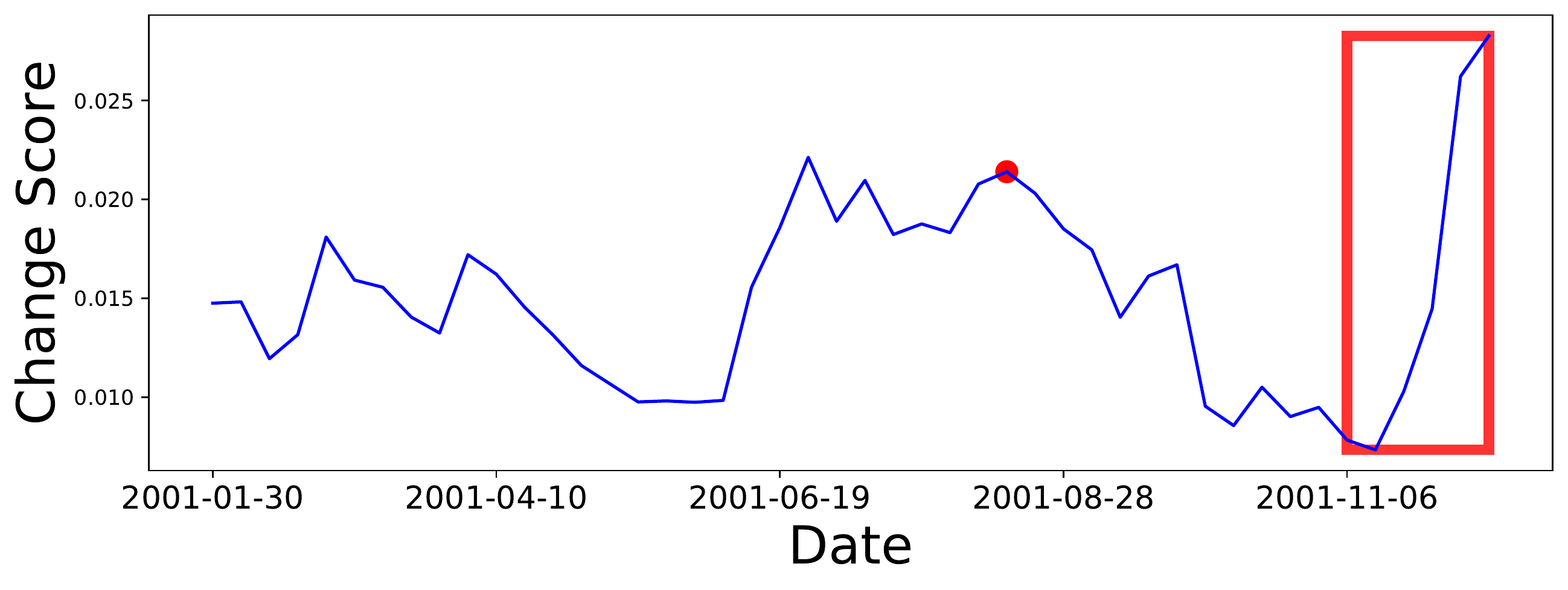}} \\
\vspace{-0.1in}
\subfloat[Karen Danne had became more active after the company declared bankruptcy, reflecting on her email length and number]{\label{fig:ar4}\includegraphics[scale=0.23]{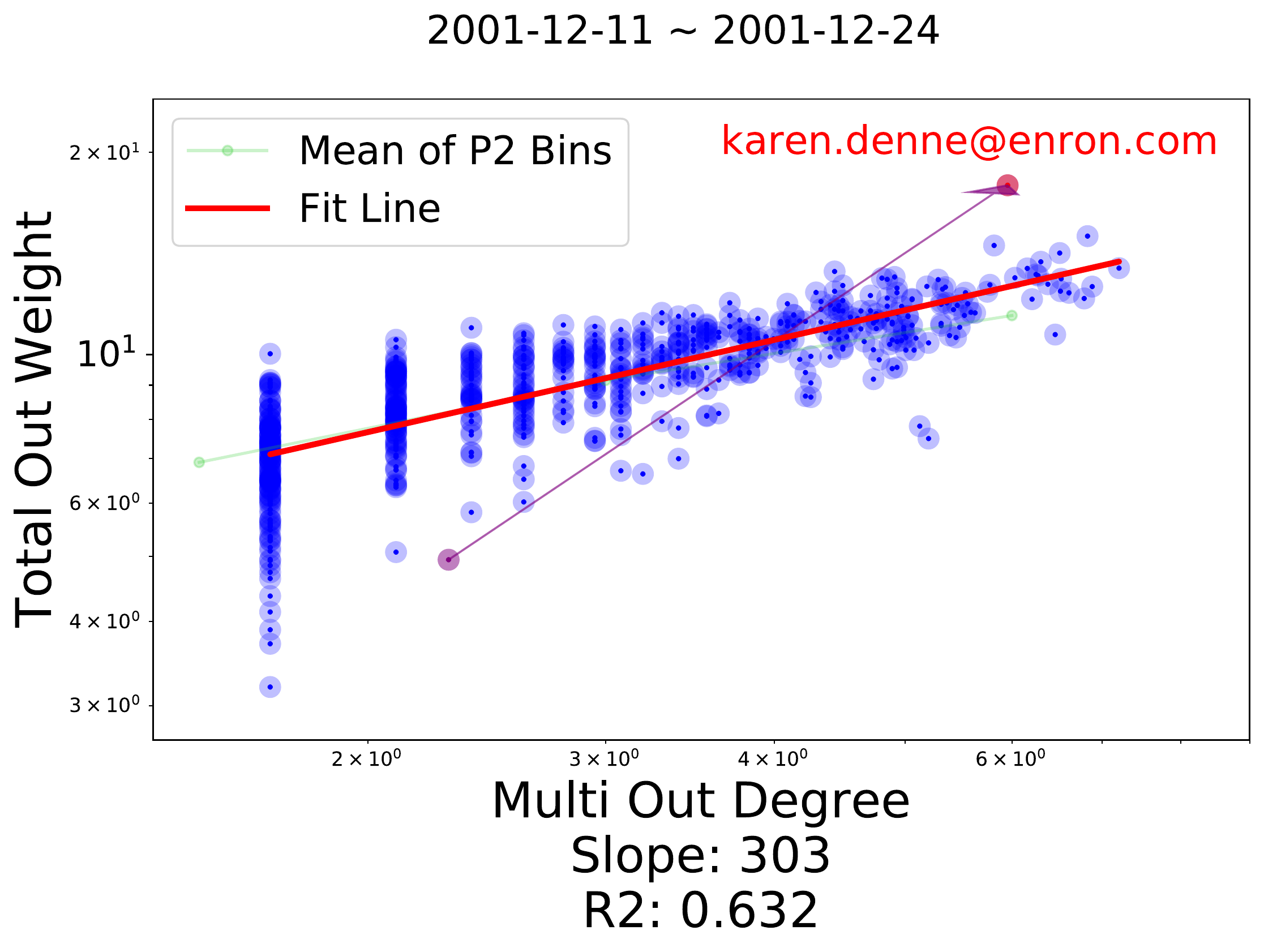}}
\caption{{\bf Attention Routing}: In Enron dataset, we capture the change when the CEO returned with company's bankruptcy}
\vspace{-0.15in}
\label{fig:ARR1}
\end{figure}

\subsubsection{Scalability}
To evaluate the scalability of \methodfor, we empirically vary the number of focus-plots and the dimension of sketch. As shown in Figure~\ref{fig:aaarsca1} and Figure~\ref{fig:aaarsca2} with black triangles, \methodfor scales linearly with both features while only using 1 machine; the efficiency boosts up after parallelizing it by 4 machines, shown with red squares.
Furthermore, we vary the number of machines to show that it is easy to parallelize and for 4 times speedup by 8 machines.




\vspace{-0.15in}
\begin{figure}[!htb]
\centering
\subfloat{\label{fig:aaarsca1}\includegraphics[scale=0.27]{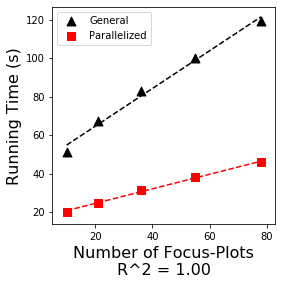}}
\subfloat{\label{fig:aaarsca2}\includegraphics[scale=0.27]{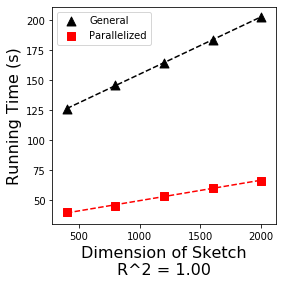}}
\subfloat{\label{fig:aaarsca3}\includegraphics[scale=0.37]{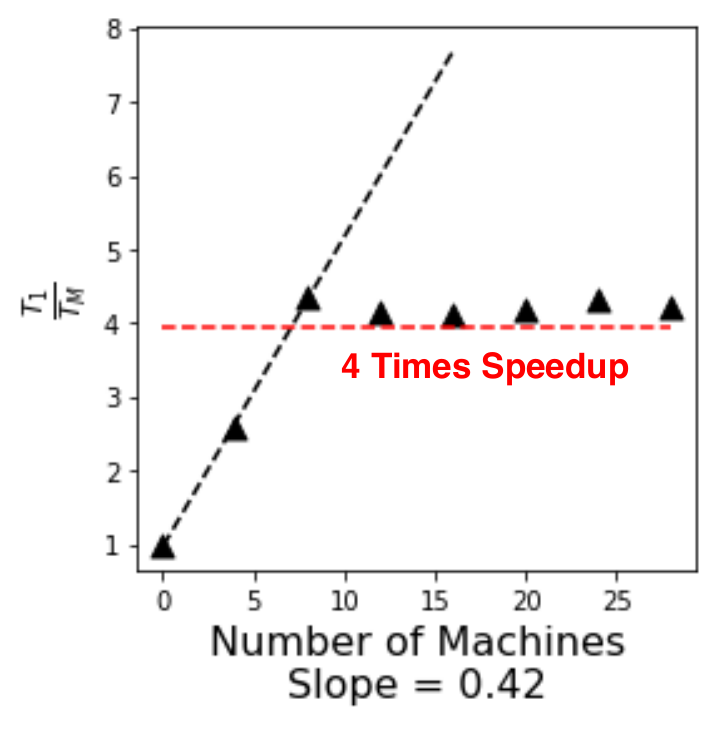}}
\caption{{\bf \methodfor}: It scales linearly with number of focus-plots (left) and dimension of sketch (middle) by one machine, and speeds up if running with more than 1 machine (right).}
\label{fig:aaarsca}
\end{figure}

\section{Conclusions}
\label{sec:concl}


We present \method (AA), which addresses the anomaly detection problem on time-evolving accounting datasets. This kind of data is usually complicated and hard to organize. Our main purpose is to automatically spot anomalies, such as money laundering, providing huge convenience for auditors and risk management professionals. Our approach is also general enough to be easily modified to solve problems in different domains. The following are our major contributions:
\ben
\item \methodsec, a parameter-free and unsupervised method to detect ``Smurfing'' type money laundering.
\item \methodfor, attend to the most suspicious plots and accounts in time-evolving graphs.
\item \methodthi, discover month-pairs with high correlation, proved by ``success stories'', and patterns in accounting dataset follow Power-Laws in log-log scales.
\item Scalability with running time linear in input size, and generality on other real-world graph datasets.
\een

\vspace{-0.02in}


\section*{Acknowledgment}
This research was partially supported by Ministry of Science and Technology Taiwan under grant no. 109-2218-E-009-014.

\bibliographystyle{IEEEtran}
\bibliography{BIB/ref}

\begin{thebibliography}{10}
\providecommand{\url}[1]{#1}
\csname url@samestyle\endcsname
\providecommand{\newblock}{\relax}
\providecommand{\bibinfo}[2]{#2}
\providecommand{\BIBentrySTDinterwordspacing}{\spaceskip=0pt\relax}
\providecommand{\BIBentryALTinterwordstretchfactor}{4}
\providecommand{\BIBentryALTinterwordspacing}{\spaceskip=\fontdimen2\font plus
\BIBentryALTinterwordstretchfactor\fontdimen3\font minus
  \fontdimen4\font\relax}
\providecommand{\BIBforeignlanguage}[2]{{%
\expandafter\ifx\csname l@#1\endcsname\relax
\typeout{** WARNING: IEEEtran.bst: No hyphenation pattern has been}%
\typeout{** loaded for the language `#1'. Using the pattern for}%
\typeout{** the default language instead.}%
\else
\language=\csname l@#1\endcsname
\fi
#2}}
\providecommand{\BIBdecl}{\relax}
\BIBdecl

\bibitem{chadha2018handling}
A.~Chadha and P.~Kaur, ``Handling smurfing through big data,'' in \emph{Big
  Data Analytics}.\hskip 1em plus 0.5em minus 0.4em\relax Springer, 2018, pp.
  459--470.

\bibitem{paula2016deep}
E.~L. Paula, M.~Ladeira, R.~N. Carvalho, and T.~Marzag{\~a}o, ``Deep learning
  anomaly detection as support fraud investigation in brazilian exports and
  anti-money laundering,'' in \emph{ICMLA}.\hskip 1em plus 0.5em minus
  0.4em\relax IEEE, 2016.

\bibitem{savage2016detection}
D.~Savage, Q.~Wang, P.~Chou, X.~Zhang, and X.~Yu, ``Detection of money
  laundering groups using supervised learning in networks,'' \emph{arXiv
  preprint arXiv:1608.00708}, 2016.

\bibitem{liflowscope}
X.~Li, S.~Liu, Z.~Li, X.~Han, C.~Shi, B.~Hooi, H.~Huang, and X.~Cheng,
  ``Flowscope: Spotting money laundering based on graphs.''

\bibitem{reuter2003money}
P.~Reuter and E.~M. Truman, ``Money laundering: Methods and markets,''
  \emph{Chasing Dirty Money: The Fight Against Money Laundering. Peterson
  Institute}, 2003.

\bibitem{madinger2011money}
J.~Madinger, \emph{Money laundering: A guide for criminal investigators}.\hskip
  1em plus 0.5em minus 0.4em\relax CRC Press, 2011.

\bibitem{corselli2020italy}
L.~Corselli, ``Italy: money transfer, money laundering and intermediary
  liability,'' \emph{Journal of Financial Crime}, 2020.

\bibitem{schneider2004money}
S.~Schneider, \emph{Money laundering in Canada: an analysis of RCMP
  cases}.\hskip 1em plus 0.5em minus 0.4em\relax Nathanson Centre for the Study
  of Organized Crime and Corruption Toronto, 2004.

\bibitem{ranshous2015anomaly}
S.~Ranshous, S.~Shen, D.~Koutra, S.~Harenberg, C.~Faloutsos, and N.~F.
  Samatova, ``Anomaly detection in dynamic networks: A survey,'' \emph{Wiley
  Interdisciplinary Reviews: Computational Statistics}, vol.~7, no.~3, 2015.

\bibitem{eswaran2018spotlight}
D.~Eswaran, C.~Faloutsos, S.~Guha, and N.~Mishra, ``Spotlight: Detecting
  anomalies in streaming graphs,'' in \emph{KDD}, 2018, pp. 1378--1386.

\bibitem{DBLP:conf/kdd/ManzoorMA16}
E.~A. Manzoor, S.~M. Milajerdi, and L.~Akoglu, ``Fast memory-efficient anomaly
  detection in streaming heterogeneous graphs,'' in \emph{{KDD}}.\hskip 1em
  plus 0.5em minus 0.4em\relax {ACM}, 2016, pp. 1035--1044.

\bibitem{akoglu2010event}
L.~Akoglu and C.~Faloutsos, ``Event detection in time series of mobile
  communication graphs,'' in \emph{Army science conference}, vol.~1, 2010.

\bibitem{eswaran2018sedanspot}
D.~Eswaran and C.~Faloutsos, ``Sedanspot: Detecting anomalies in edge
  streams,'' in \emph{ICDM}.\hskip 1em plus 0.5em minus 0.4em\relax IEEE, 2018,
  pp. 953--958.

\bibitem{levi2006money}
M.~Levi and P.~Reuter, ``Money laundering,'' \emph{Crime and Justice}, vol.~34,
  no.~1, 2006.

\bibitem{DBLP:conf/sigmod/MatsubaraSF14}
Y.~Matsubara, Y.~Sakurai, and C.~Faloutsos, ``Autoplait: automatic mining of
  co-evolving time sequences,'' in \emph{{SIGMOD} Conference}.\hskip 1em plus
  0.5em minus 0.4em\relax {ACM}, 2014.

\bibitem{rissanen1978modeling}
J.~Rissanen, ``Modeling by shortest data description,'' \emph{Automatica},
  vol.~14, no.~5, 1978.

\bibitem{liu2008isolation}
F.~T. Liu, K.~M. Ting, and Z.-H. Zhou, ``Isolation forest,'' in
  \emph{ICDM}.\hskip 1em plus 0.5em minus 0.4em\relax IEEE, 2008.

\bibitem{DBLP:journals/corr/abs-1710-05333}
N.~Gupta, D.~Eswaran, N.~Shah, L.~Akoglu, and C.~Faloutsos, ``Lookout on
  time-evolving graphs: Succinctly explaining anomalies from any detector,''
  \emph{CoRR}, vol. abs/1710.05333, 2017.

\bibitem{berka1999workshop}
\BIBentryALTinterwordspacing
P.~Berka, ``Workshop notes on discovery challenge pkdd'99,'' 1999. [Online].
  Available: \url{http://lisp.vse.cz/pkdd99/}
\BIBentrySTDinterwordspacing

\bibitem{klimt2004enron}
B.~Klimt and Y.~Yang, ``The enron corpus: A new dataset for email
  classification research,'' in \emph{ECML}.\hskip 1em plus 0.5em minus
  0.4em\relax Springer, 2004, pp. 217--226.

\bibitem{chau2012data}
D.~H. Chau, ``Data mining meets hci: Making sense of large graphs,'' CMU
  Pittsburgh PA Machine Learning Dept, Tech. Rep., 2012.

\bibitem{DBLP:conf/www/PanditCWF07}
S.~Pandit, D.~H. Chau, S.~Wang, and C.~Faloutsos, ``Netprobe: a fast and
  scalable system for fraud detection in online auction networks,'' in
  \emph{{WWW}}.\hskip 1em plus 0.5em minus 0.4em\relax {ACM}, 2007, pp.
  201--210.

\bibitem{chau2011polonium}
D.~H.~â. Chau, C.~Nachenberg, J.~Wilhelm, A.~Wright, and C.~Faloutsos,
  ``Polonium: Tera-scale graph mining and inference for malware detection,'' in
  \emph{SDM}.\hskip 1em plus 0.5em minus 0.4em\relax SIAM, 2011, pp. 131--142.

\bibitem{DBLP:conf/asunam/DevineniKFF15}
P.~Devineni, D.~Koutra, M.~Faloutsos, and C.~Faloutsos, ``If walls could talk:
  Patterns and anomalies in facebook wallposts,'' in \emph{{ASONAM}}.\hskip 1em
  plus 0.5em minus 0.4em\relax {ACM}, 2015, pp. 367--374.

\end{thebibliography}

\newpage
\appendix



\subsection{Background and Related Work}
\label{sec:background}




\subsubsection{Anti-Money Laundering (AML)}
Money laundering has been a long-term trouble to the modern financial system. Many scholars and researchers have tried to combat ``Smurfing''. As shown in Figure ~\ref{fig:smurfing}, it contains three major stages: placement, layering and integration. Placement denotes injecting a huge amount of money from illegal activities to financial systems; layering separates the money into several parts to avoid Suspicious Activity Report. Finally, these money will be integrated into the target account.

On one hand, some studies used transaction-based AML to detect ``Smurfing''. \cite{chadha2018handling} used Hadoop MapReduce to efficiently examine the transactions, and transactions that exceeds the fine-tuned thresholds will be detected. \cite{paula2016deep} used AutoEncoder to calculate the reconstruction error for anomaly detection. The higher the error, the more suspicious the transaction. However, these methods ignore the dependence between accounts, so that they are more likely to trigger false alarms.

On the other hand, considering transactions as a graph is a perfect solution for clarifying interactions between accounts. \cite{savage2016detection} combined network analysis with supervised learning to detect the groups of money laundering activities. Nevertheless, these methods highly rely on carefully labelled data, where the availability is low and the cost is significant. Thus, \cite{liflowscope} proposed to approach AML as an optimization problem by using multipartite graphs. Nevertheless, it requires a lot of senders and receivers. To ameliorate the aforementioned problems, we propose a more practical definition, and carry out the result with unsupervised method.

\subsubsection{Anomaly Detection in Time-Evolving Graphs}
Anomaly detection is one of the closest real-world applications in graph mining, which has been well-studied in the past few years. Time-evolving graphs are dynamic and sudden changes are usually inklings of potential anomalies. Careful scrutiny of this field can be found in \cite{ranshous2015anomaly}. In addition, \cite{eswaran2018spotlight} proposed a randomized sketching-based approach to detect sudden changes in streaming graphs, and theoretically proved that anomalies lie far apart from normal instances. StreamSpot \cite{DBLP:conf/kdd/ManzoorMA16} is a clustering-based method that detects anomalies in streaming heterogeneous graphs with dynamic maintenance. An algorithm is proposed to detect changing points in time-evolving graphs by the time series feature extracted from all nodes \cite{akoglu2010event}. \cite{eswaran2018sedanspot} proposed their method based on two suggested signs of anomalous edges; sudden occurrence and connection with sparsely-connected parts.

\subsubsection{Attention Routing}
This term was first coined in \cite{chau2012data}, which is used to increase the quality of analyzing massive networks. The major purpose is to find the most critical nodes and subgraphs. \cite{DBLP:conf/www/PanditCWF07} designed a system to detect fraud users in a large online auction dataset, modeling users and transactions as a Markov Random Field. \cite{chau2011polonium} detected the malware in a tera-scale bipartite graph based on scalable Belief Propagation algorithm. LookOut \cite{DBLP:journals/corr/abs-1710-05333} enhanced the interpretability and succinctness, only showing a limited number of plots which maximize the anomaly score.



\subsection{\methodthi for Insight Discovery} \label{app:num2}
Effective audits require auditors to identify risks through understanding an entity---its environment, operation, and control activities.
The real-world challenge has been the lack of fast tools to process the massive amount of detailed transaction data and establish robust patterns regarding an entity's business operations and accounting practices.
Quickly discovering periodic temporal correlations in large accounting data is profoundly meaningful because it aids auditors in better accessing risk and detecting aberrant transactions.

\subsubsection{Problem Definition}
Temporal correlations usually exist in a time-evolving graph dataset, which are worth studying because they are useful in discovering the behavior rules of nodes. The nodes that exhibit unusual behavior can be treated as outliers. The distributions of different features in the dataset also give us abundant information. Here we define the problem in the following:
\begin{definition}[Insight Discovery]
Given a time-evolving graph dataset, there are two major goals to discover the insights of dataset:
\bit
\item unearth the correlation between time frames and the most significant factor incurring this phenomenon, and
\item examine the distributions to verify whether they follow Power-Law patterns.
\eit
\end{definition}

\subsubsection{Information Extraction and Visualization}  \label{ssec:tra}
SVD is a method of matrix decomposition which is extremely useful in extracting essential information. It decomposes the matrix into $U$, the columns containing left singular vectors, $\Sigma$ with singular values on its diagonal, and $V^T$ containing right singular vectors. To further investigate the correlation between months, we form a 2-dimensional matrix as shown in Figure~\ref{fig:mmm}, where the first dimension is the pairs of the source and destination accounts, and the second dimension is months. We use SVD to decompose this matrix to derive three matrices $U$, $\Sigma$, and $V^T$. The rows of $V^T$ illustrate correlations between each month, and the corresponding columns of $U$ explain the correlations according to the behaviors of each pair. Therefore we can project rows and columns on a 2-dimensional plane, and the lower the distance between the elements of rows and columns, the higher the relevance. This visualization can be used to discover active months of each account pair.

\begin{figure}[hbtp]
\centering
\includegraphics[width=0.3\textwidth]{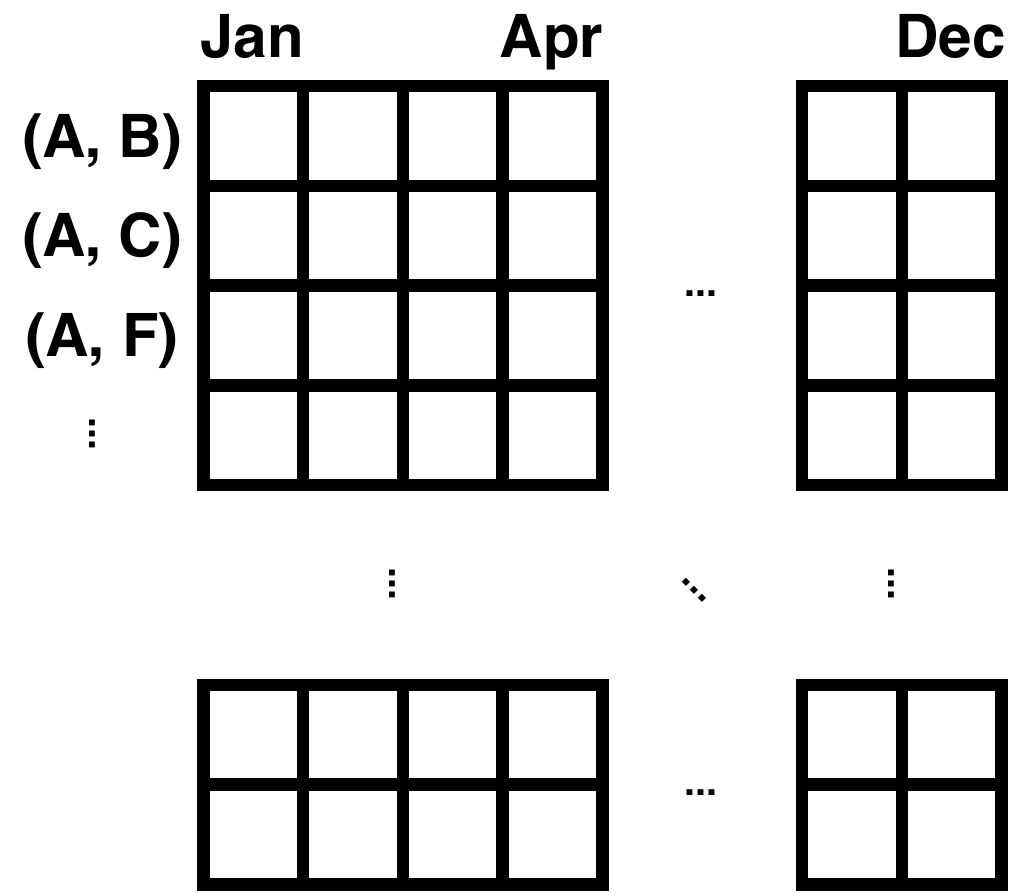}
\caption{A 2-dimension matrix where the first dimension is the pair of accounts, and the second dimension is months}
\label{fig:mmm}
\end{figure}

\subsubsection{Power-Law Pattern} \label{par:powerlaw}
Log-logistic distribution has been well studied and used in economics for many years, which reflects how the ``rich get richer.''
A continuous random variable follows the log-logistic distribution, if and only if its logarithm follows the logistic distribution \cite{DBLP:conf/asunam/DevineniKFF15}. One of the best properties of the log-logistic is that its odds ratio follows the Power-Law, which is the ratio of Cumulative Distribution Function (CDF) over the complementary CDF (CCDF). We use Linear Regression to fit odds ratio and calculate the slope $\rho$. We then derive the equation as follows:
\begin{equation}
    \beta = -\rho
\end{equation} \label{equ:oddsratio}
where $\beta$ denotes the slope of CCDF and $\rho$ denotes the slope of odds ratio. This equation is used to examine whether a distribution follows the Power-Law in the experiments.

\subsection{Additional Experiment Results for ``Smurfing'' Detection}
To implement \methodsec on streaming graphs, we further calculate the compression rate by encoding cost over time. As shown in Figure~\ref{fig:cr}, \methodsec handle streaming graphs well, where it successfully detects ``Smurfing'' pattern by giving the highest compression rate over time.

\begin{figure}[hbtp]
\centering
\includegraphics[width=0.45\textwidth]{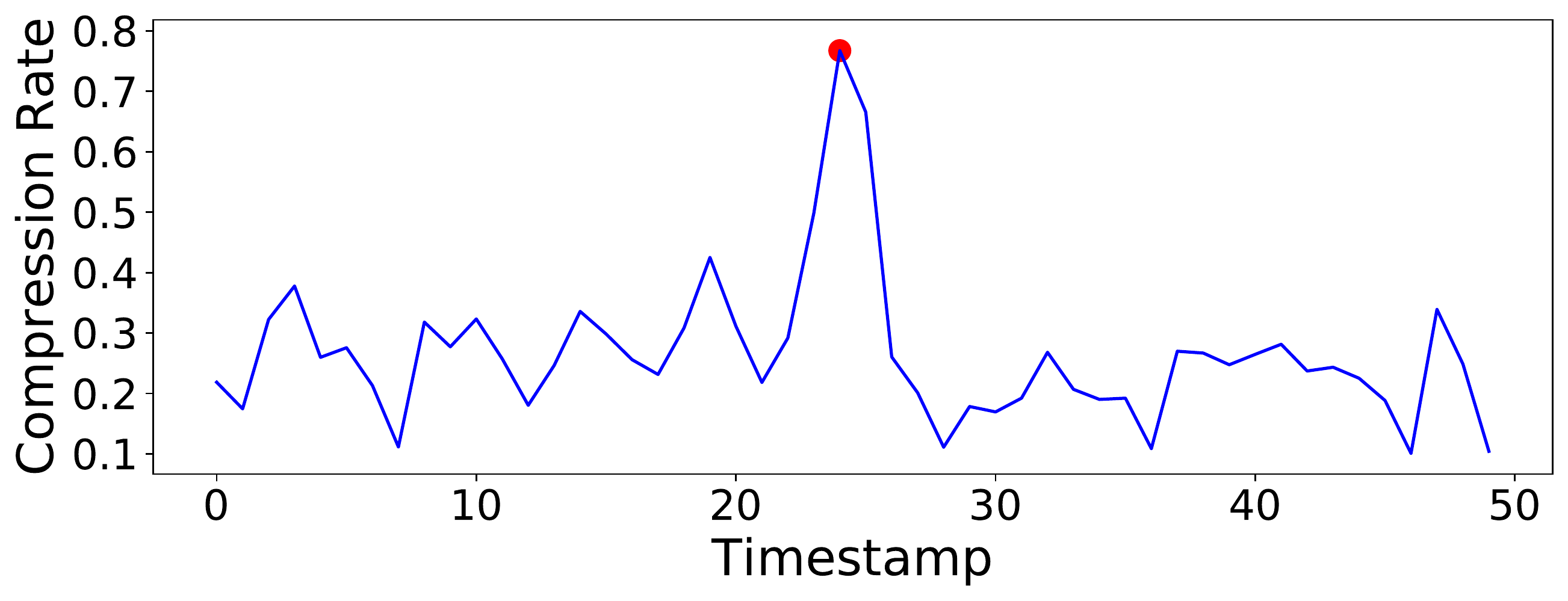}
\caption{{\bf \methodsec}: \methodsec detects ``Smurfing'' in the time-evolving dataset on the red point with high compression rate}
\label{fig:cr}
\end{figure}

\subsection{Additional Experiment Results for Attention Routing} \label{app:ar}
We use \methodfor to investigate activities in the Accounting dataset. As shown in Figure~\ref{fig:ar1}, the change score reaches the highest peak at the detected red point in the last week of April. Figure~\ref{fig:ar2} further shows during this focal week, the change of the Cost of Good Sold (COSG) account 44810 is significant in terms of median and total out weight. The COSG account moves from the purple dot to the red dot, where the purple dot summarizes the account behavior during the four weeks before the focal week, and the red dot describes the account behavior in the focal week. During previous four weeks, the median and total out weight do not equal zero (purple dot), which indicates a small movement of the COGS account in terms of the account balance. The median and total out weight then increase dramatically in the focal week (red dot), suggesting the occurrence of exceptionally significant movement concerning the COGS account balance. According to Figure~\ref{fig:ar2}, similar implications can be drawn for the Finished Goods Inventory account 13381.

\begin{figure}
\centering
\subfloat[The change score over time, where the significant change is shown in red point]{\label{fig:ar1}\includegraphics[scale=0.23]{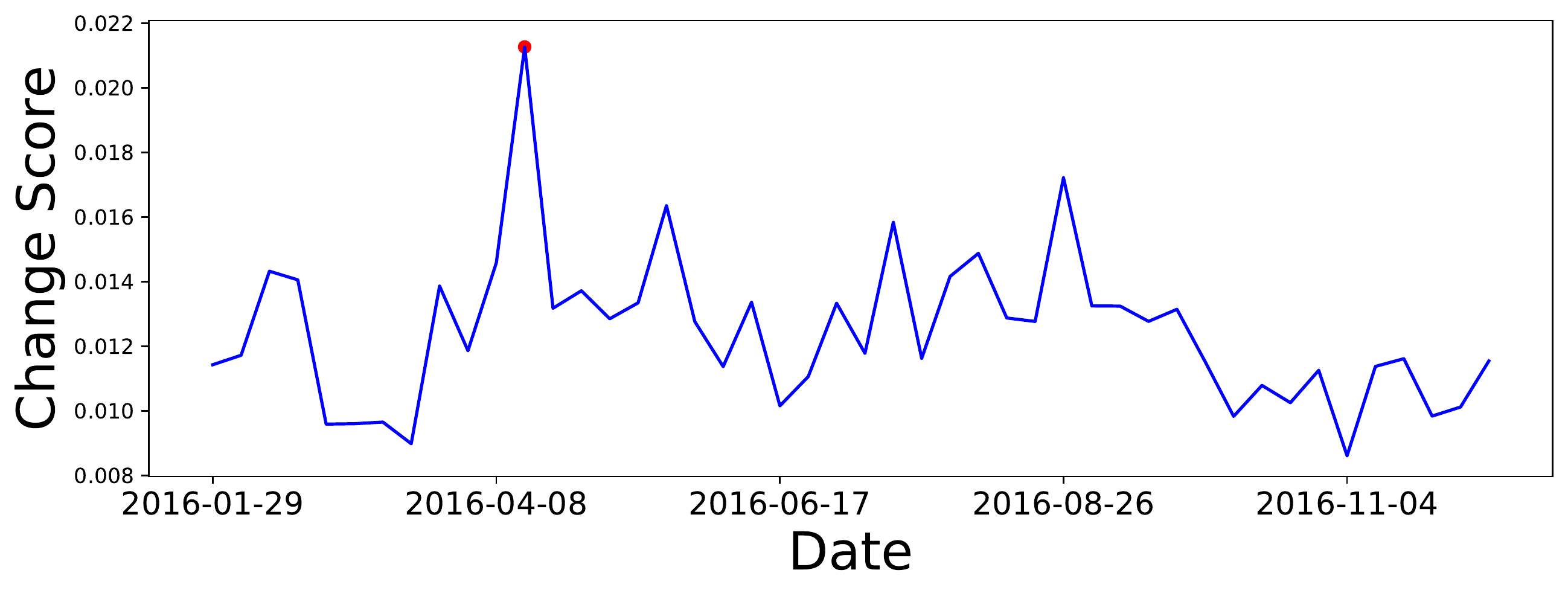}} \\
\subfloat[The explanation where account 44810 and 13381 had significant changes in terms of median and total out weight]{\label{fig:ar2}\includegraphics[scale=0.23]{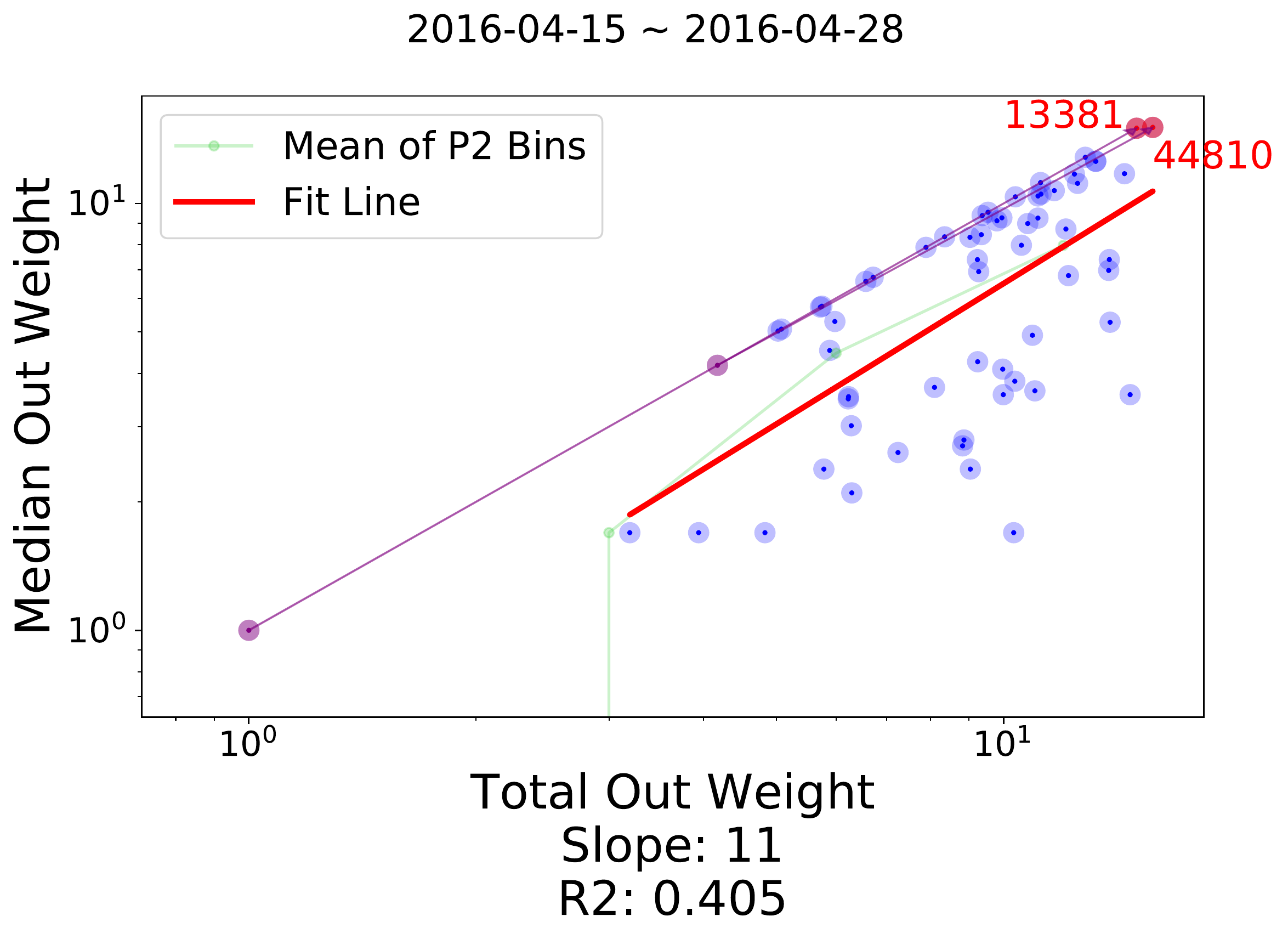}}
\caption{{\bf Attention Routing}: In Accounting Dataset, we discovery manufacturer's alternation of information system}
\label{fig:ARR2}
\end{figure}

\subsection{Additional Experiment Results for Insight Discovery} \label{app:id}
In this experiment, we selected the data from a general ledger database of an anonymous manufacturer in 2016.

\subsubsection{Temporal Correlation}
In order to further illustrate the factors causing this phenomenon, we extract the information by SVD described in Section~\ref{ssec:tra}. As shown in Figure~\ref{fig:svd}, $V^T_1$ and $V^T_2$ rows reflect the similarities between each month-pairs. We then project them into a 2-d plane marked by red triangles as shown in Figure~\ref{fig:cluster}. To further illustrate the relation between month-pairs and account-pairs, we project the corresponding $U_1$ and $U_2$ columns marked by dots on the same plane, where the colors are generated by K-Means. The clusters around month-pairs are the activated account-pairs in corresponding month-pairs, which we call them ``petals''. For example, the purple cluster are the account-pairs that were active only in January and June. The yellow cluster on the left of the plane denotes those account-pairs being in two month-pairs, i.e., (Jan, Jun) and (Mar, Oct). Moreover, the cluster in the middle of the plane is called ``stamen'', which contains months with no correlation. These ``petals'' and ``stamen'' give us extremely intuitive information to investigate correlations between months.

\begin{figure}[!t]
\centering
\subfloat[The rows of matrix $V$ containing right singular vectors]{\label{fig:svd}\includegraphics[scale=0.3]{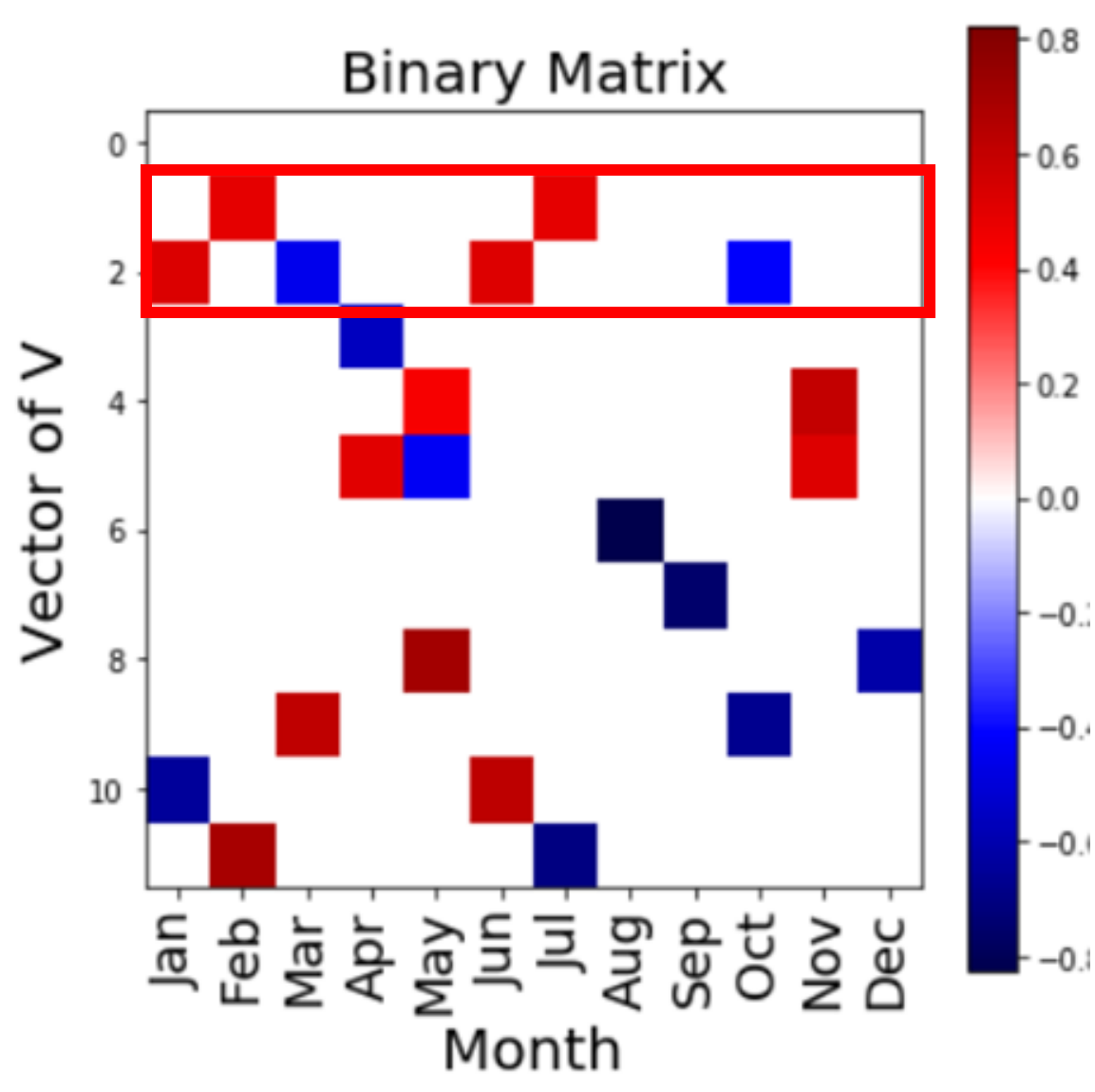}}\quad
\subfloat[The account pairs are gathered into clusters on $U_1$ and $U_2$ plane ]{\label{fig:cluster}\includegraphics[scale=0.3]{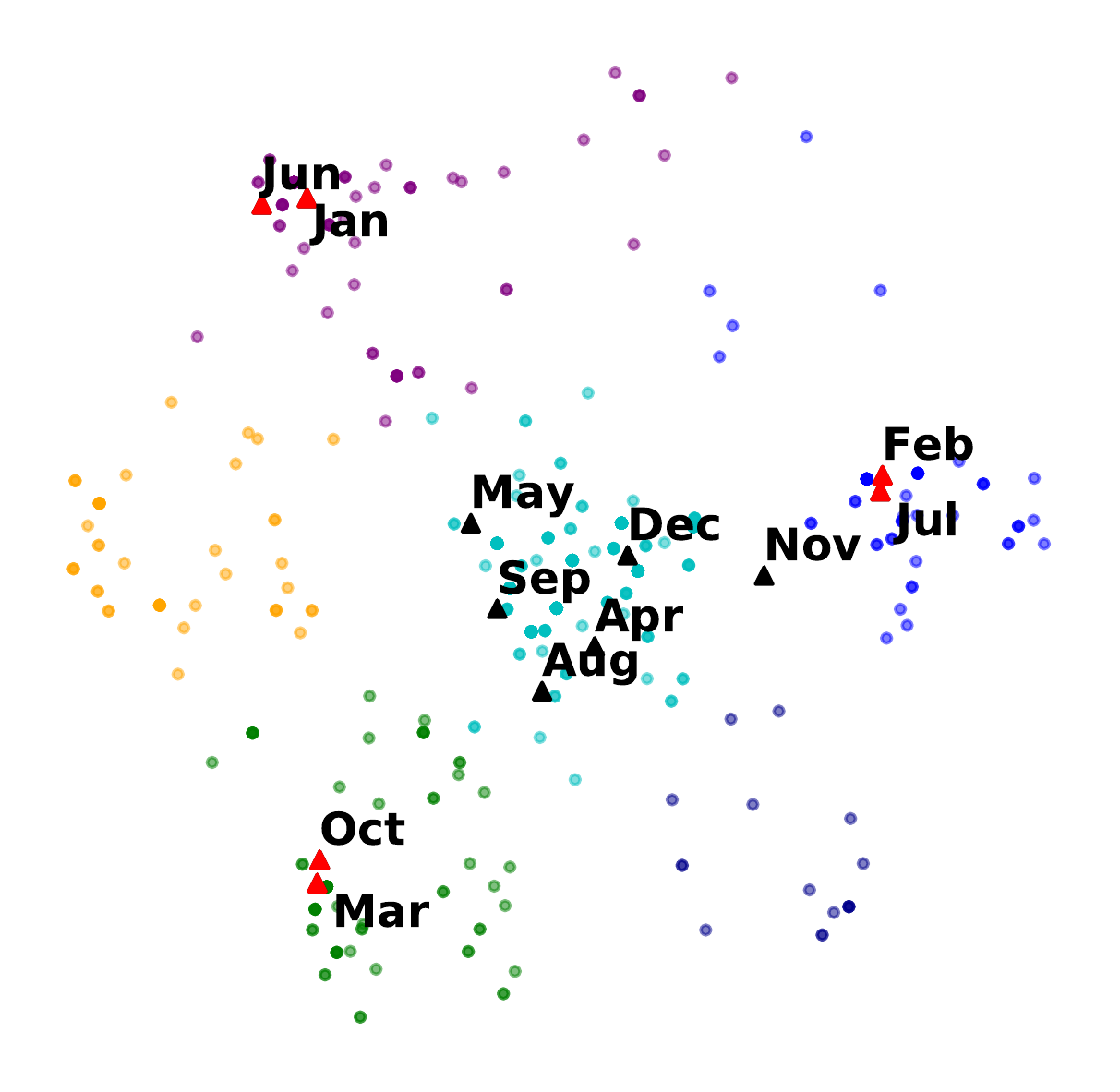}}
\caption{{\bf Temporal Correlation}: We discover the temporal correlation between months}
\label{fig:RAeffective}
\end{figure}

The anonymous manufacturer we investigate conducted onboard training of new employees in March and October. Typically, such events are held in the selected months of the year (and perhaps during the downtime of the employees affected). When such events happen, certain account-pairs are activated, such as the cost of issuing an award to the best-performing employees and a specific funding source is usually used for these HR events. As a result, the account-pair from the funding source to the award expense only takes place in the two months of the year. Meanwhile, this manufacturer set its intercompany reconciliations in March and October, possibly because of the relatively more intense waves of intercompany transactions in these months. These waves of intercompany transactions and scheduled reconciliations would trigger the pairings between intercompany accounts.
As shown in Figure~\ref{fig:cluster}, March and October are clustered as a similar month-pair.

\subsubsection{Patterns Following Power-Law} \label{sssec:pfpl}
In the accounting dataset, we draw several distributions of CCDF and Odds Ratio and surprisingly find that log-logistic distribution fits well. The features we extracted for computing distribution are as follows:
\begin{enumerate*}[label=(\roman*)]
  \item Transaction Amount
  \item Interval Arriving Time
  \item Pair Amount
  \item Pair Multiplicity
  \item Out Weight / 1000
  \item In Weight / 1000
  \item Multi Out Degree
  \item Multi In Degree
  \item Unique Out Degree
  \item Unique In Degree.
\end{enumerate*}
We use account-pair as the unit of distributions (\romannumeral 2) and (\romannumeral 3). Weight is divided by 1000 to eliminate small numbers. Multi degree means that the degree to one specific node could be more than one, where in a unique degree there will only be one.
As shown in Table~\ref{tbl:dist}, the results demonstrate the accuracy of Equation~\ref{equ:oddsratio}. Moreover, there are three essential exponents 1, 0.6, and 1.6, and all distributions in the Accounting dataset follow these three exponents. For distributions from (\romannumeral 1) to (\romannumeral 2), $\rho$ usually equals to 1; for the ones from (\romannumeral 3) to (\romannumeral 8), $\rho$ usually equals to 0.6; for the others from (\romannumeral 9) to (\romannumeral 10), $\rho$ usually equals to 1.6.

\begin{table*}[htbp]
    \begin{center}
    \begin{tabular}{ | c | c |}
    \hline
    \multicolumn{2}{| c |}{\textcolor{red}{$\beta = \rho \approx 1$}} \\
    \hline
    Transaction Amount & Interval Arriving Time \\
    \hline
    \includegraphics[scale=0.18]{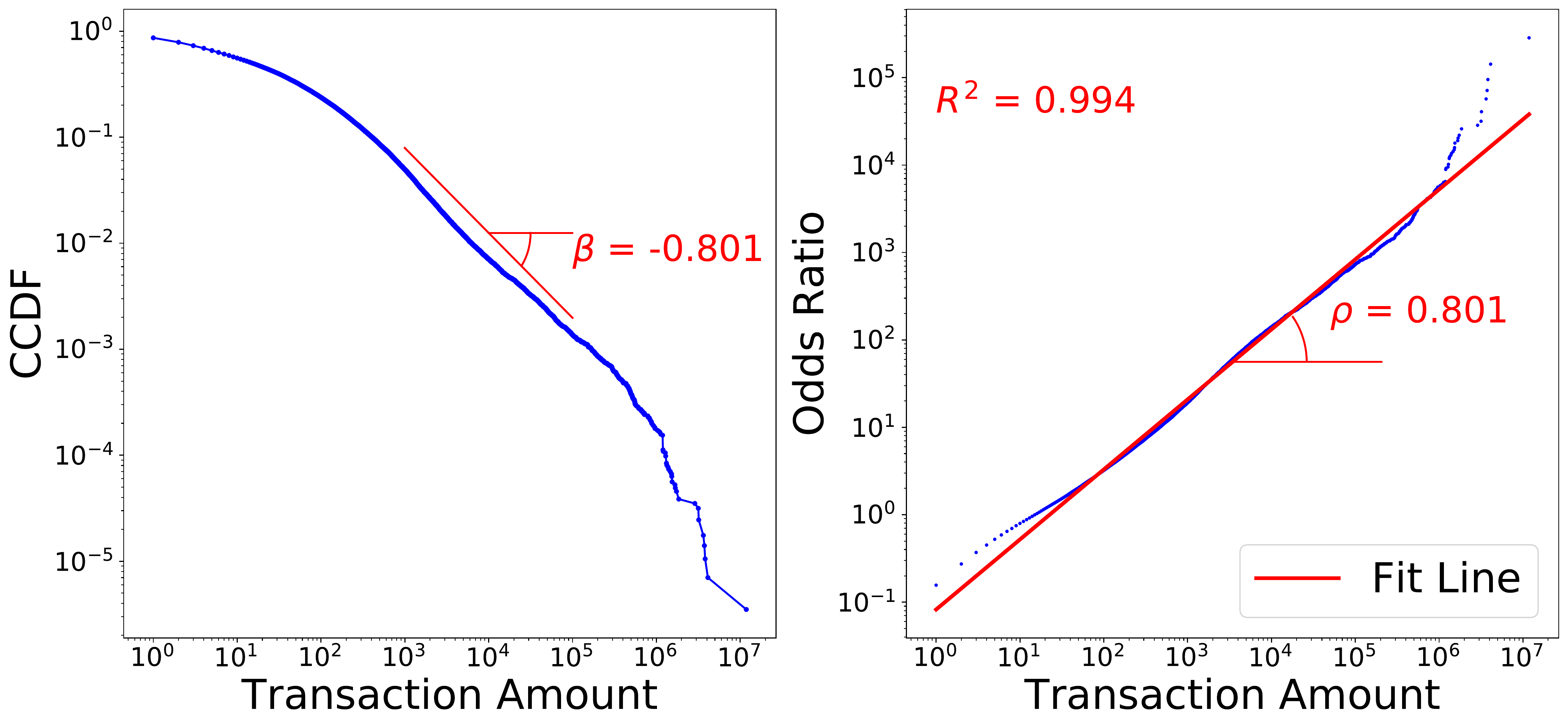}
    &
    \includegraphics[scale=0.18]{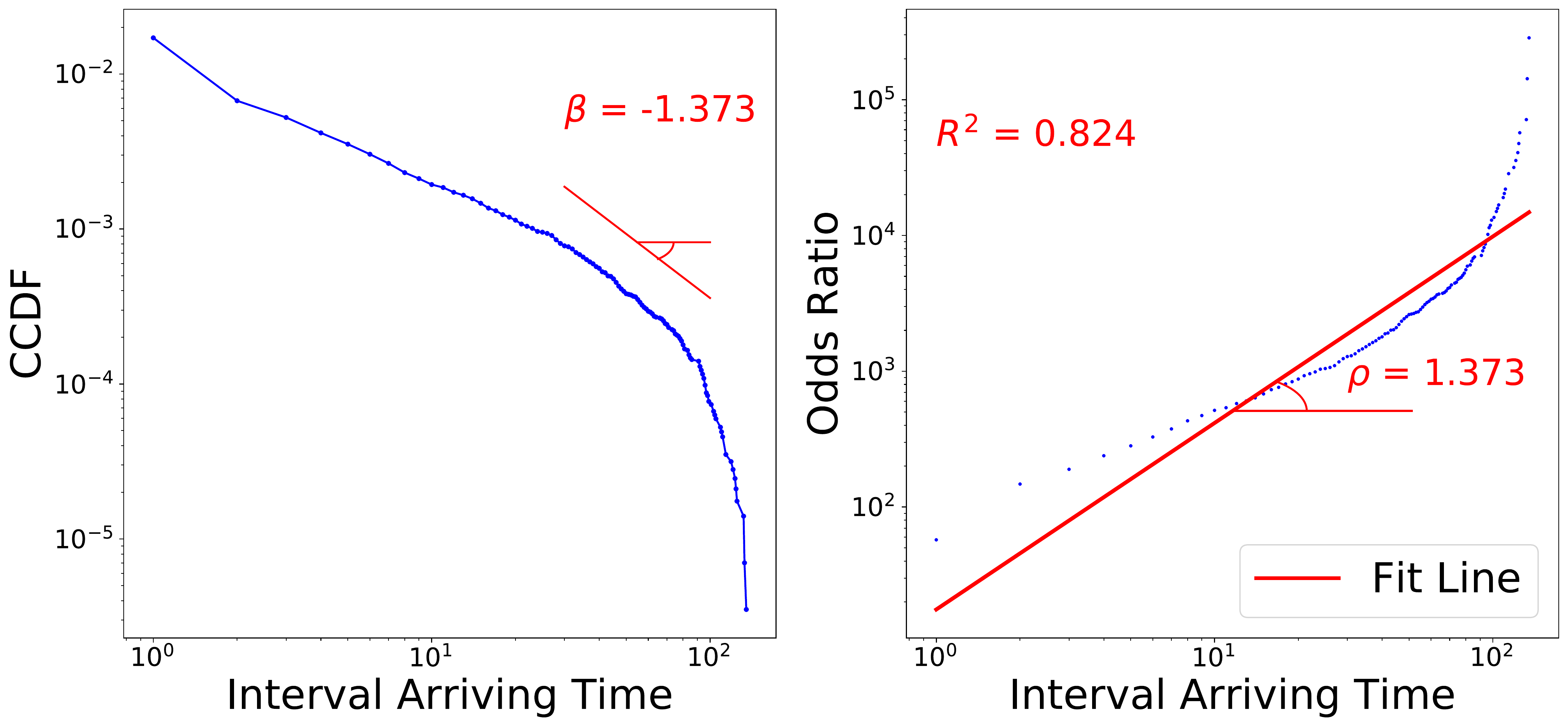}
    \\ \hline
    \multicolumn{2}{| c |}{\textcolor{red}{$\beta = \rho \approx 0.6$}} \\ 
    \hline
    Pair Amount & Pair Multiplicity \\
    \hline
    \includegraphics[scale=0.18]{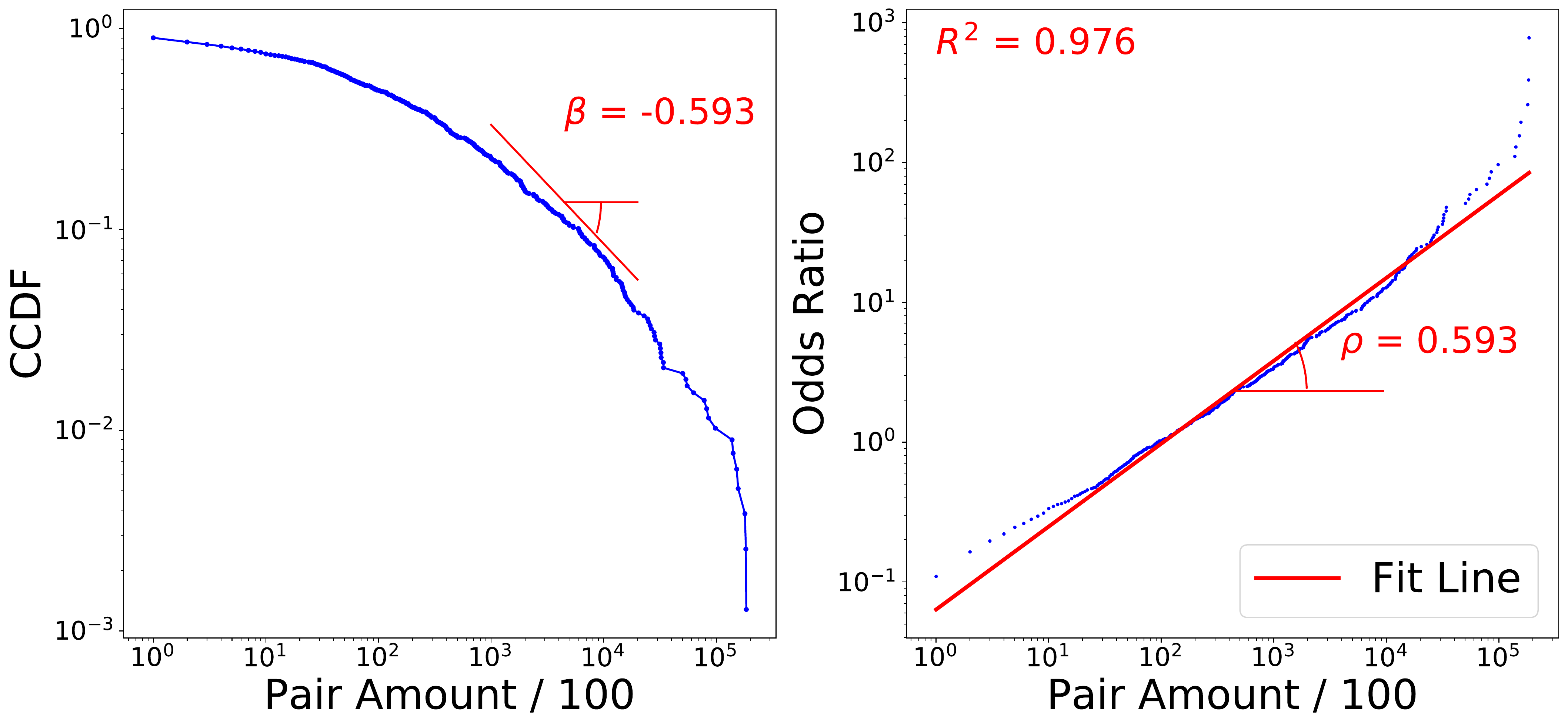}
    &
    \includegraphics[scale=0.18]{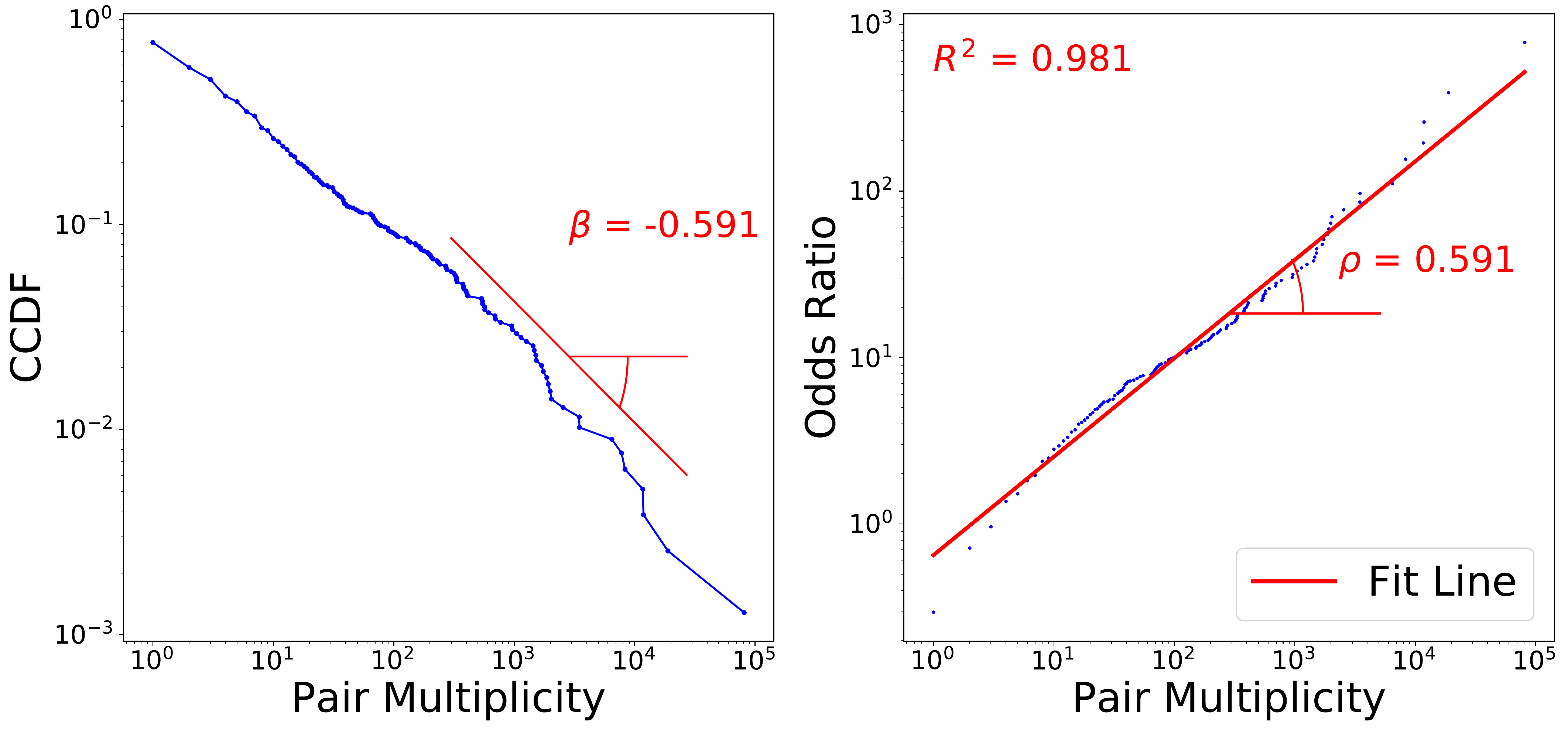}
    \\ \hline
    Out Weight / 1000 & In Weight / 1000 \\
    \hline
    \includegraphics[scale=0.18]{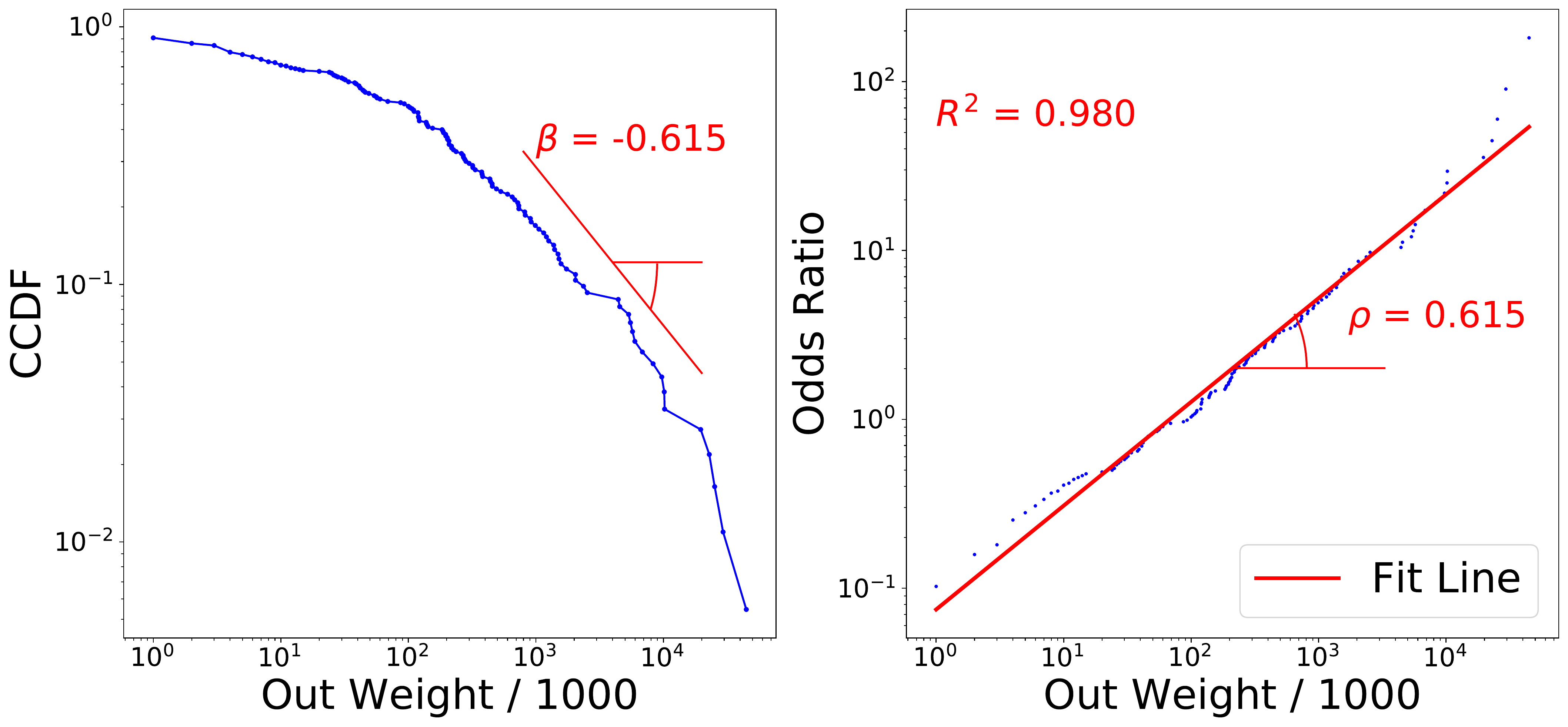}
    &
    \includegraphics[scale=0.18]{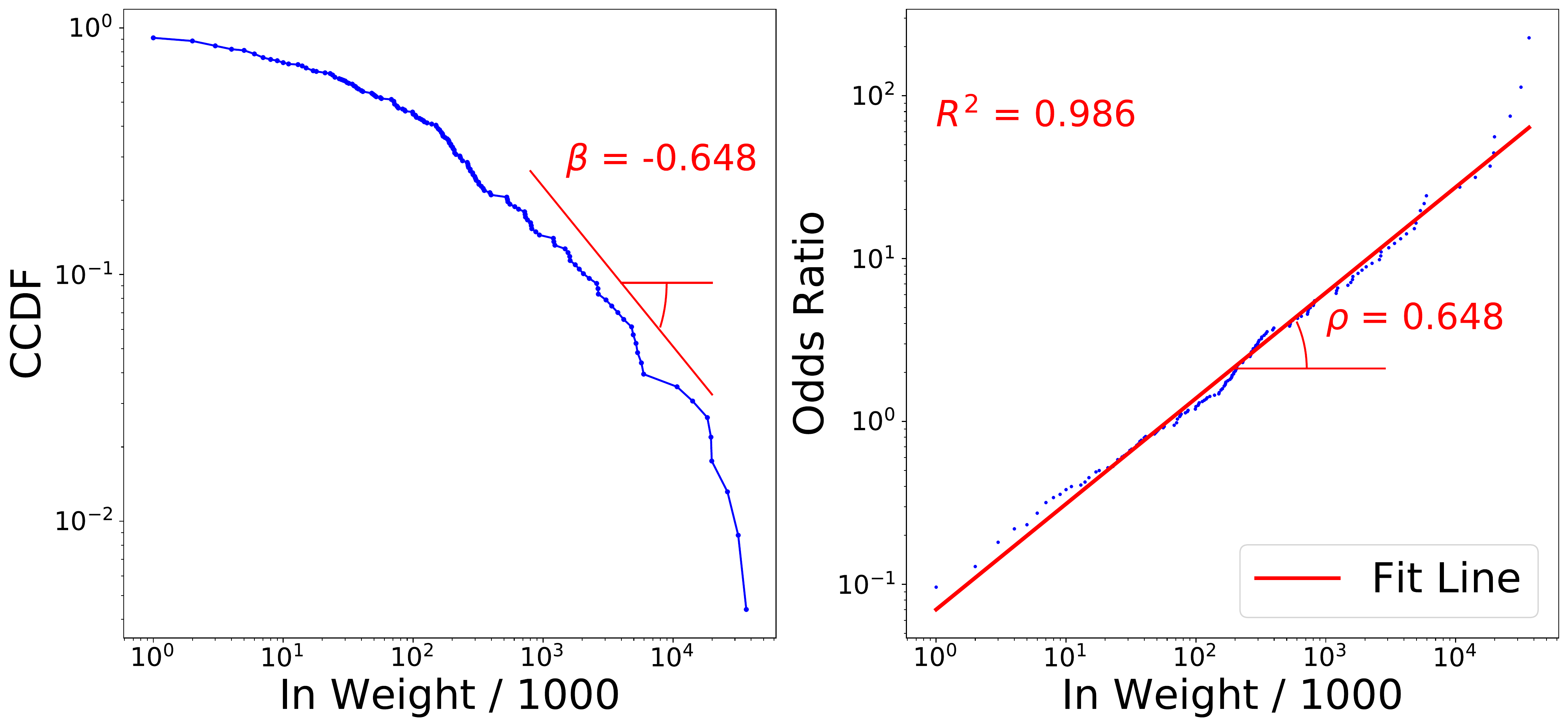}
    \\ \hline
    Multi Out Degree & Multi In Degree \\
    \hline
    \includegraphics[scale=0.18]{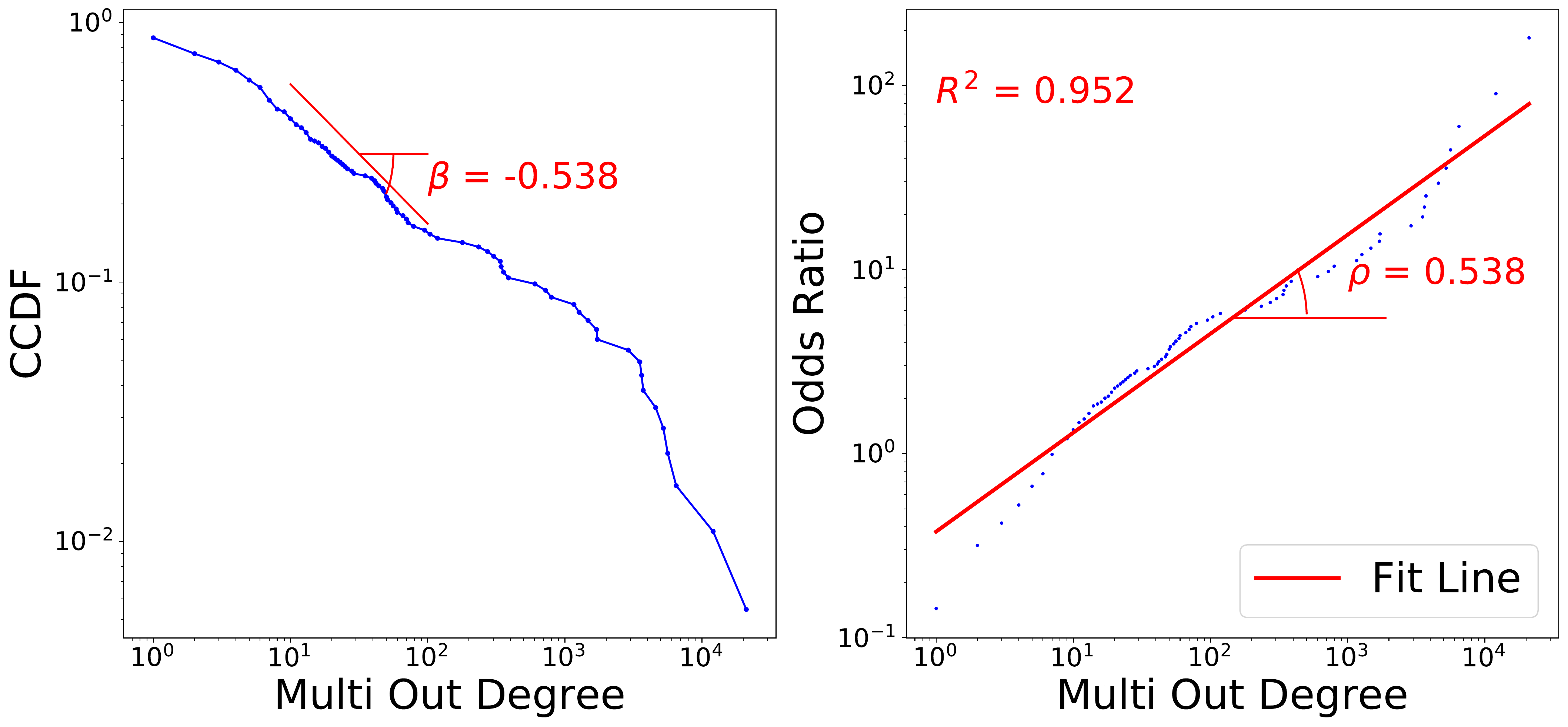}
    &
    \includegraphics[scale=0.18]{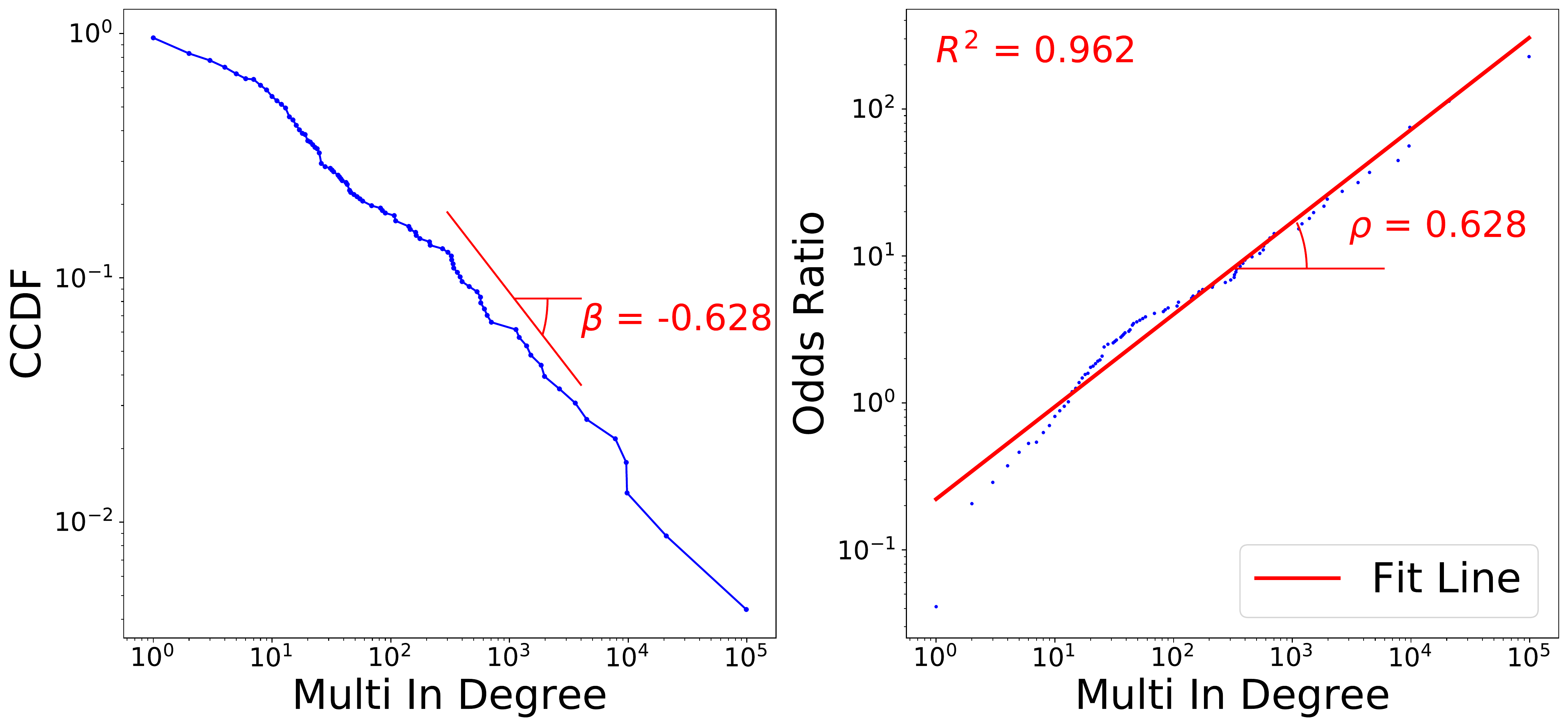}
    \\ \hline
    \multicolumn{2}{| c |}{\textcolor{red}{$\beta = \rho \approx 1.6$}} \\ 
    \hline
    Unique Out Degree & Unique In Degree \\
    \hline
    \includegraphics[scale=0.18]{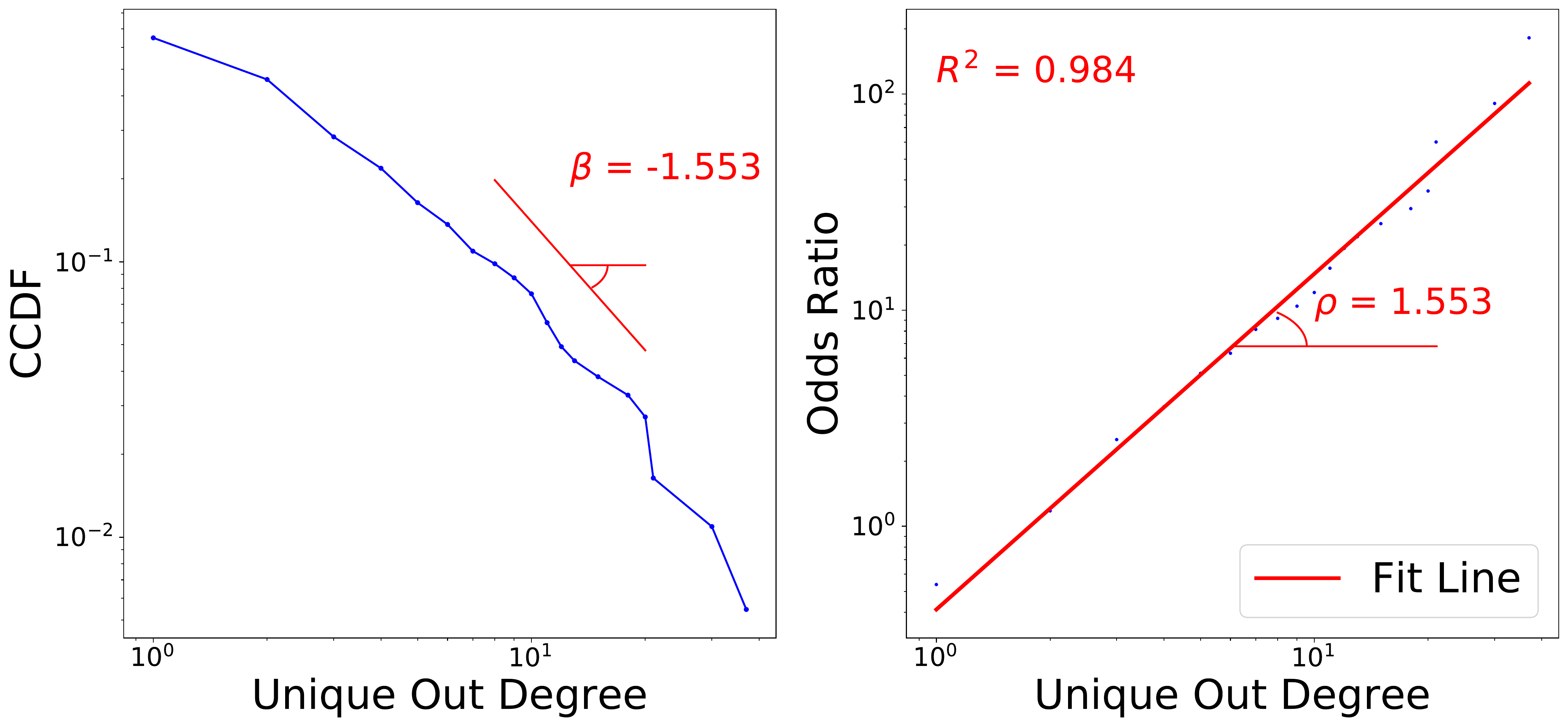}
    &
    \includegraphics[scale=0.18]{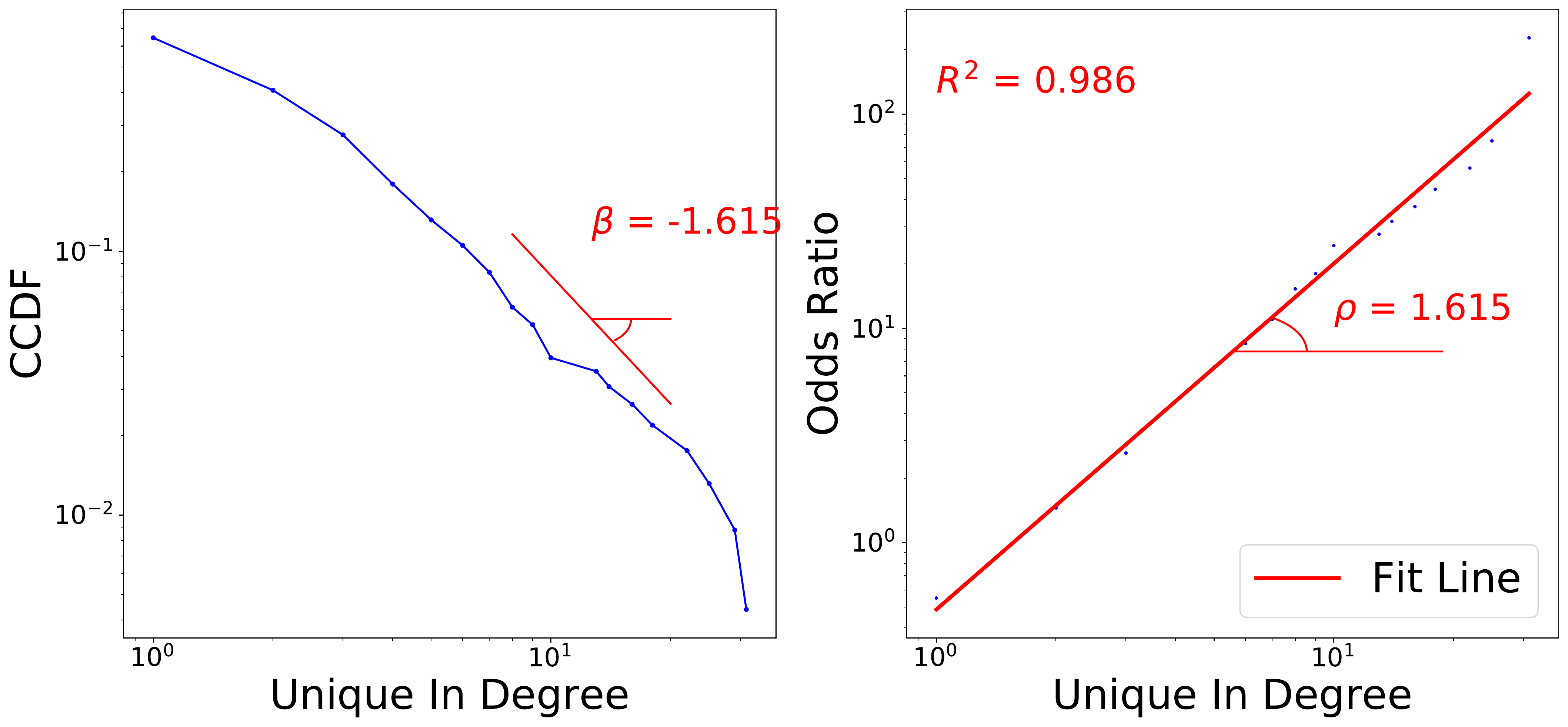}
    \\ \hline
    \end{tabular}
    \caption{{\bf \methodthi}: \methodthi discovers the patterns in accounting dataset follow Power-Law.}
    \label{tbl:dist}
    \end{center}
\end{table*}

\end{document}